\documentclass[12pt]{article}
\usepackage[top=1.25in,bottom=1in,left=1in,right=1in]{geometry}
\usepackage[comma,authoryear]{natbib} 
\usepackage{cite}
\usepackage{tikz}
\usepackage{amssymb}

\usepackage{amsmath,amsfonts,eurosym,geometry,ulem,graphicx,caption,color,setspace,sectsty,comment,footmisc,caption,natbib,pdflscape,subfigure,array}
\usepackage{authblk}

\usepackage[colorlinks,urlcolor=blue,citecolor=blue]{hyperref}

\newtheorem{theorem}{Theorem}
\newtheorem{observation}{Observation}
\newtheorem{lemma}{Lemma}
\newtheorem{definition}{Definition}
\newtheorem{claim}{Claim}
\newtheorem{corollary}{Corollary}
\newtheorem{proposition}{Proposition}
\newtheorem{assumption}{Assumption}
\newenvironment{proof}[1][Proof]{\noindent\textbf{#1.} }{\ \rule{0.5em}{0.5em}}

\newcolumntype{L}[1]{>{\raggedright\let\newline\\arraybackslash\hspace{0pt}}m{#1}}
\newcolumntype{C}[1]{>{\centering\let\newline\\arraybackslash\hspace{0pt}}m{#1}}
\newcolumntype{R}[1]{>{\raggedleft\let\newline\\arraybackslash\hspace{0pt}}m{#1}}

\title{The Limits of Search Algorithms\thanks{We are deeply grateful to Wei Zhao for his support and guidance throughout this work. For helpful comments, we thank Zhonghong Kuang, Sanxi Li, Yichuan Lou, Daisuke Oyama, Takuo Sugaya, Jidong Zhou, the participants of the RUC Theory Reading Group, and the attendees at GAMES 2024.}
\vspace{-0.5em}}

\author{
    Xiaoyu Chen\thanks{Graduate School of Economics, University of Tokyo. Email: lzdcxy@gmail.com}, \ \
    Jingmin Huang\thanks{School of Economics, Renmin University of China. Email: jingmin.huang@ruc.edu.cn},  \ \
    Yibo Lian\thanks{Department of Economics, Pennsylvania State University. Email: yilian.econ@gmail.com}
}

\date{\today\vspace{-0.5em}}

\begin{document}
\maketitle
\begin{abstract}
A platform commits to a search algorithm that maps prices to search order. Given this algorithm, sellers set prices, and consumers engage in sequential search. This framework generalizes the ordered search literature. We introduce a special class of search algorithms, termed “contracts,” show that they implement all possible equilibrium prices and then characterize the set of implementable prices. Within this set, we identify the seller-optimal contract, whose first-best outcome remains an open problem for a multiproduct seller. Our findings highlight the conditions under which the platform favors price dispersion or price symmetry. Furthermore, we characterize the consumer-optimal and socially optimal contracts, which exert opposing forces to the seller-optimal contract: while the seller-optimal contract promotes higher prices, the consumer-optimal and socially optimal contracts favor lower prices.

\end{abstract}

\section{Introduction}

Consumers shop and sequentially search for the best alternative on platforms such as Amazon or eBay. These platforms can influence consumers’ search behavior—and consequently, sellers’ sales and welfare distribution—by designing product rankings through their search algorithms, as the order in which products are visited affects sales.

The ordered search literature examines market equilibrium under a given search algorithm, which can be classified into two categories. The first category assumes an exogenous search order, as in \citet{wolinsky_true_1986}, where searches are random, and \citet{armstrong_prominence_2009}, where one seller is prominent. The second category consists of price-directed search, building on the optimal search rule of \citet{weitzman_optimal_1979}, as seen in \citet{HAAN2018223} and \citet{choi_consumer_2018}. Specifically, under symmetric sellers, the optimal search rule dictates that consumers inspect products in ascending order of prices.

We generalize these two categories by modeling the search algorithm as a function that maps prices to search orders, an approach also adopted by \citet{BarIsaac2023Monetizing}. In particular, we consider a model with one consumer, one monopolistic platform, and two sellers. The platform first commits to a search algorithm, after which sellers set prices and are ranked accordingly. Finally, the consumer begins her shopping and search process. Within this framework, we investigate: \textit{(i) which prices and search orders are implementable, and (ii) among the set of implementable prices and search orders, which ones are optimal under different objective functions.}




We address question (i) by introducing a special class of search algorithms, which we term “contracts.” A contract is a tuple that specifies sellers’ prices and search orders while penalizing non-compliant sellers by ranking them second. We demonstrate that any equilibrium in the economy can be implemented by a contract, reminiscent of the revelation principle. Consequently, we focus on contracts without loss of generality and characterize the set of implementable prices by deriving sellers’ incentive compatibility constraints for contract acceptance. Specifically, for a contract to be accepted, no seller should find unilateral deviation profitable. We show that the set of implementable prices is compact, convex, and exchangeable. Regarding comparative statics, higher search costs increase the penalty of being ranked last, strengthening sellers’ incentives to adhere to the platform’s suggested prices, thereby expanding the set of implementable prices.

For question (ii), we first examine industry profit maximization, which corresponds to platforms’ common proportional commission fee revenue model. We also explore an open question in the literature: how should a platform price and position its products if it operates as a multiproduct seller, i.e., when it is integrated with sellers? We show that in this case, the platform optimally makes one product prominent and sets higher prices for it. When the platform and sellers are not integrated, the maximization problem is subject to sellers’ incentive compatibility constraints. The optimal search algorithm depends on the magnitude of search costs: if search costs exceed a certain threshold, the optimal search algorithm resembles the first-best outcome, featuring asymmetric prices and a higher probability of ranking the higher-priced seller first. However, if search costs fall below the threshold, the optimal search algorithm instead induces the highest possible price for each seller while implementing random search.

Next, we identify the optimal contracts for maximizing (or minimizing) total demand, social welfare, and consumer surplus. The contract that sets the lowest symmetric price and implements random search maximizes total demand, social welfare, and consumer surplus, while simultaneously minimizing industry profit. In other words, under the proportional commission fee revenue model, the platform’s incentives are misaligned with those of society: the platform’s optimal contract requires high prices, whereas low prices are necessary to enhance social welfare. This result also suggests that the misalignment can be resolved if the platform’s revenue is based on sellers’ sales rather than their profits.



\paragraph{Related Literature}As mentioned above, the ordered search literature examines market equilibrium under a specific search algorithm, where the search order is either exogenous \citep{wolinsky_true_1986, arbatskaya_ordered_2007, armstrong_prominence_2009} or determined by \citet{weitzman_optimal_1979}’s optimal search rule \citep{HAAN2018223, choi_consumer_2018}. Our paper develops a more general framework for search algorithms, incorporating these existing models as special cases.

Among studies on search algorithms, \citet{BarIsaac2023Monetizing} is the most closely related to our work. They consider the same formulation of a search algorithm as a mapping from prices to search orders and analyze the optimal steering method and fee structure, comparing ad auctions with search algorithms. However, their model assumes that consumers observe only the prominent seller rather than engaging in sequential search. In contrast, our paper adopts a sequential search framework, enabling a more nuanced analysis and direct comparison with the sequential search literature.

Our work also relates to the platform design literature, including \citet{hagiu2011intermediaries} and \citet{zhong_platform_2022}. Studies in this area typically assume that platforms possess superior information compared to consumers and examine the platform’s role and incentives in improving matching efficiency. In contrast, our paper assumes that the platform does not have more information than consumers. Instead, we focus on the platform’s role in allocating prominence among sellers.

In Section \ref{sec: model}, we describe the setup and compare contracts with search algorithms used in the sequential search literature. In Section \ref{sec:implementable prices}, we address question (i) by characterizing the set of implementable prices. In Section \ref{sec: optimization}, we answer question (ii) by analyzing optimal contracts under different objective functions. In Section \ref{sec: discussion}, we discuss corner solutions.

\section{The Model}
\label{sec: model}
The environment consists of two sellers \footnote{Like \citet{BarIsaac2023Monetizing}, this paper adopts a parsimonious model that considers only two sellers. Explicitly modeling additional sellers would significantly increase the complexity of the problem, as a search algorithm would need to specify the probability of each possible permutation of the  $n$ -seller rankings.}, labeled $i = 1,2$, one platform, and one unit mass of consumers, where each consumer demands one unit product. Each seller \(i\) sells a single product on the platform, with a price \(p_{i}\) and a constant marginal cost \(c_{i}\) normalized to zero. Consumers have idiosyncratic match utilities from seller $i$'s product, denoted by $u_i$, which is not observed by sellers. A particular consumer anticipates that seller $i$'s match utility is independently and identically drawn from a common distribution with CDF $F(u)$ and PDF $f(u)$, whose support is normalized to $[0,1]$. 
Therefore, a consumer who eventually buys from seller $i$ obtains payoff $u_i - p_i$. We assume \(p_{i}\) is observed by consumers before they search, and to know $u_i$, consumers can incur a search cost $s$ to inspect seller $i$. After the search, consumers purchase the product with the highest payoff $u_i - p_i$ among all inspected products or take the outside option, which is normalized to zero. We assume free recall so consumers have no return cost to a previously inspected product. The first search is assumed to be free, so consumers search at least one product.

The platform determines the search order by committing to a ``search algorithm" $\alpha_{i}: \mathbb{R}^2 \to \mathbb{R}, (p_1,p_2)\mapsto \alpha_i(p_1,p_2)$ for each seller. A search algorithm maps prices to the probability that each seller is searched first. \footnote{Another interpretation of $\alpha_i$ is the proportion of consumers that first inspect seller $i$, which is implicitly adopted by the literature of ordered search, e.g. \citet{armstrong_ordered_2017}.} So \(\alpha_i \in [0,1]\) and \(\sum_i{\alpha_i} =1 \). When there are only two sellers, we denote \(\alpha_1\) by \(\alpha\) and $\alpha_2$ by $1-\alpha$ to simplify our notations. 


Different ranks affect sellers’ demand (elaborated further in the next subsection) and their pricing decisions. Let $D_i^n$ represent seller $i$'s demand when ranked in the $n$-th position. Each seller's profit under search algorithm $\alpha(p_1,p_2)$ is the weighted sum of profits when ranked first and second:
\begin{equation}
    \begin{aligned}
    &\Pi_1(p_1,p_2) = \alpha(p_1,p_2) p_1D_1^1(p_1,p_2) + (1-\alpha(p_1,p_2) ) p_2 D_1^2(p_1,p_2)\\
    &\Pi_2(p_1,p_2) =  \alpha(p_1,p_2)  p_2D_2^2(p_1,p_2) + (1-\alpha(p_1,p_2) ) p_2D_2^1(p_1,p_2)
\end{aligned} 
\label{eq: profits}
\end{equation}

The game unfolds as follows: First, the platform commits to a ranking function $\alpha$. Next, the sellers set their prices to maximize their respective profits, $\Pi_1$ and $\Pi_2$. Finally, consumers engage in a sequential search process. We focus on pure strategy Nash equilibria of sellers' pricing decisions.




\subsection{Sequential Search}
For a sequential search problem, \citet{weitzman_optimal_1979}'s seminal results outline the \textit{optimal selection rule} and the \textit{optimal stopping rule}. In our context, although we let the platform determine the search order, the optimal selection rule is not violated since the platform also makes the first search free. So consumers are willing to follow the search order prescribed by the platform. 

The optimal stopping rule features a cutoff structure. Suppose consumers first search seller $i$ and learn $u_i$, if \( u_{i} \leq p_{i} \), the consumer never buys from seller $i$. If she takes the outside option, her surplus is zero. If she incurs a search cost \( s \) to visit seller \(j\),  she gains \(u_{j}-p_{j}\) when \(u_{j}-p_{j} \geq 0\) and takes the outside option when \(u_{j}-p_{j}<0\). Define 
\[V(p) \equiv \mathbb{E}_{u}[\max \{0, u-p\}],\nonumber\]
so that \( V(p_{j})-s \) is the expected surplus of the search.

Define \(V(A) \equiv s\), so the consumer will never search if and only if the expected surplus of search is negative, i.e., \(V(p_{j})-s<0\) or seller \(j\) charges \(p_{j}>A\)\footnote{\(V(p) = \int_p^1 (u-p) dF(u)\) is strictly decreasing since \(V^{\prime}(p) =F(p)-1 \leq 0,\) and equality holds if and only if \(p=1\).}. Therefore, \(A\) is the threshold price that induces no search when the consumer cannot benefit from the first seller. The seller ranked second will never charge a price above \(A\).

When \( u_{i}>p_{i} \), the consumer's surplus from seller \(i\) is \( u_{i}-p_{i} \). If she searches for seller \( j \), her expected surplus is
\[\mathbb{E}_{u_{j}}[\max \{u_{i}-p_{i}, u_{j}-p_{j}\}]-s=u_{i}-p_{i}+V(p_{j}+u_{i}-p_{i})-s,\]
so the consumer will buy from seller $i$ immediately without sampling seller $j$ if and only if \(V(p_{j}+u_{i}-p_{i})-s < 0\), i.e., \( u_{i}-p_{i} > A-p_{j} \). Here, \(A\) is the threshold match utility that induces immediate purchase when the two sellers offer the same price. Specifically, when the second seller charges above \(A\), the consumer will buy from the first seller immediately.

Figure \ref{Demand} shows the demand distribution when seller \(i\) is ranked first. The closed form of demand can be found in the Appendix \ref{appendix: search problem}.  

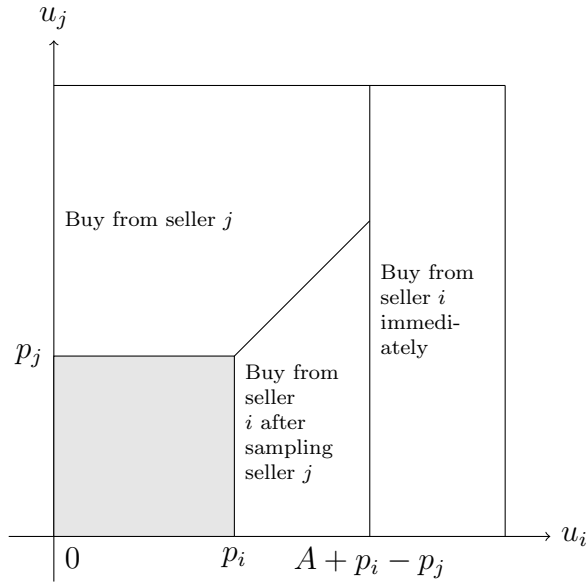
\begin{figure}[h]
\centering
\caption{Distribution of demand when seller $i$ is ranked first.}
   \begin{tikzpicture}[scale=6]
    \draw[->] (-0.1,0) -- (1.1,0) node[right] {$u_i$} node[below right] at (0,0) {$0$} ;
    \draw[->] (0,-0.1) -- (0,1.1) node[above] {$u_j$} node[right, font = \scriptsize] at (0,0.7) {Buy from seller $j$} node[left] at (0,0.4) {$p_j$} ;
    \draw (1,0) -- (1,1);
    \draw (0,1) -- (1,1);
    \filldraw[fill=gray!20!white] (0,0) rectangle (0.4,0.4) node[right,text width = 1.3cm, font = \scriptsize] at (0.4,0.25){Buy from seller $i$ after sampling seller $j$};
    \draw  (0.7,1) -- (0.7,0) node[below]{$A + p_i - p_j$} node[midway,right,text width = 1.3cm, font = \scriptsize]{Buy from seller $i$ immediately};
    \draw (0.4,0.4)  -- (0.7,0.7);
    \draw (0.4,0) node[below] {$p_i$};
    \end{tikzpicture}
    \label{Demand}
\end{figure} 

The distribution of demand reveals that the first-searched (prominent) seller benefits from search frictions. In the literature, this advantage is attributed to differences in the price quotes received by consumers \citep{bergemann_search_2021}. Without search frictions ($s = 0$), consumers can compare $u_1 -p_1$ and $u_2-p_2$ before making a purchase, as every consumer receives two price quotes. With search frictions $(s>0)$, those who do not search only receive price quotes from the prominent seller. When search frictions are substantial such that $A > p_j$, if we suppose $i$ is prominent, our previous analysis suggests that all consumers receive price quotes only from seller $i$, which effectively makes seller $i$ a monopolist. This advantage is represented by higher demand for the prominent seller, which in our notations is $D_i^1 - D_j^2 \geq 0$ when $p_i = p_j$.

Nevertheless, since we allow searches to be randomized ($\alpha \in [0,1]$), a seller can be both ranked first and second with positive probability under the same prices. It motivates us to consider another form of advantage, $D_i^1 - D_i^2$. We define a new function $B_i(p_1,p_2):= D_i^1(p_1,p_2)-D_i^2(p_1,p_2)$ as the difference of demand between seller $i$ being ranked first and second under price $(p_1,p_2)$. The following lemma shows this advantage always exists and is symmetric between two sellers; neither property depends on prices. Also, intuitively, the advantage increases with $s$. 

\begin{lemma}[``Bonus'']
    For any $p_1,p_2$, 
    \begin{itemize}
        \item $B(p_1,p_2): =B_1(p_1,p_2)=B_2(p_1,p_2)\ge 0$
        \item $B(p_1,p_2)$ strictly increases with $s$.
    \end{itemize}
    \label{lem: bonus}
\end{lemma}

Proof is in Appendix \ref{appendix: search problem}.

$B$'s symmetry and non-negative properties, which do not depend on prices, are the building blocks of the following analysis of this paper, while $B$'s dependence on $s$ manifests in properties of the implementable price set. We will first see how those properties allow the platform to punish sellers by simply ranking them behind, based on which we write down a special class of search algorithms and obtain a result analogous to the ``revelation principle''.

\subsection{Search Algorithms and ``Contracts''}

As previously defined, a search algorithm maps sellers' prices to the probability each seller is searched first. Here we get a closer look at those search algorithms adopted in the literature.

\noindent\textbf{Prominence \citep{armstrong_prominence_2009}} Prominence means one seller will always be ranked first regardless of prices. If seller 1 is prominent, the search algorithm is $\alpha(p_1,p_2) = 1$. Given this search algorithm, the only difference between our model and \citet{armstrong_prominence_2009} is that in their paper, the prices are unobservable before searches. A unique equilibrium $(p_1^*,p_2^*)$ ($p_1^* < p_2^*$) exists in their paper, i.e., the prominent seller sets lower prices. While in our model, the unique equilibrium is $(\hat{p}_1^*, \hat{p}_2^*)$, where $\hat{p}_1^* > \hat{p}_2^*$.  


\noindent \textbf{Random Search \citep{wolinsky_true_1986}}
Random search means every seller has the same probability to be prominent, i.e. $\alpha(p_1,p_2) = 1/2$. In this case, a symmetric equilibrium exists.

\noindent \textbf{Price-Directed Search}. When prices are advertised, if the search costs for every seller are the same and products are ex-ante symmetric, consumers search in the ascending order of prices according to the optimal selection rule \citep{weitzman_optimal_1979}.

\[\alpha(p_1,p_2) = \begin{cases}
    1, &p_1<p_2 \\
    0, &p_1>p_2\\
    \alpha \in [0,1], &p_1 = p_2
\end{cases}\]
Because slightly undercutting a rival’s price causes a discrete jump in profit, prices will be chosen according to mixed strategies in equilibrium (see, for instance, \citet{choi_consumer_2018}). 

\noindent \textbf{``Contract"}. We introduce a special class of search algorithms as follows: 
\[\alpha(p_1, p_2) = \begin{cases}
    \alpha^*,& p_1 = p_1^*, p_2 = p_2^*\\
    1, & p_1 = p_1^*, p_2 \neq p_2^*\\
    0, &p_1 \neq p_1^*, p_2 = p_2^*\\
    \alpha \in [0,1], &p_1 \neq p_1^*, p_2 \neq p_2^*
\end{cases}\]

This search algorithm can be interpreted as the platform recommending a ``contract'' $(p_1^*, p_2^*, \alpha^*)$ to sellers. If a seller does not accept the ``contract'', it will be punished by being ranked second with probability one. Since in Nash equilibrium, we care about unilateral deviation, the value of $\alpha(p_1, p_2)$ when $p_1 \neq p_1^*$ and $p_2 \neq p_2^*$ is inconsequential, provided it remains within $[0,1]$. In the next section, we will leverage contracts to characterize the set of implementable prices. 

For sharper and tractable results, we introduce the following assumptions. 

\begin{assumption}
    $u_i$ ($i = 1,2$) is uniformly distributed.
    \label{assumption: uniform}
\end{assumption}

\begin{assumption}
    $A$ is large enough, i.e., the search cost $s$ is small enough\footnote{The exact $A$ will be specified when the assumption is applied.}.
    \label{assumption: A large}
\end{assumption}

Assumption \ref{assumption: uniform} is needed since the standard sequential search model formulated as ours is barely tractable under general assumptions about $u_i$'s distribution. A recently developed tool leverages the fact that the demand system of sequential search can be derived from a discrete choice problem \citep{armstrong_which_2015}. This allows applying results from \citet{quint_imperfect_2014} to ensure log-concavity and log-supermodularity of demand under the assumption that the CDF and survival function of a new random variable are log-concave. While \citet{choi_consumer_2018} extend this method to obtain general results for a slightly modified search problem, their approach is not applicable to our model, as the new random variable arising in our corresponding discrete choice problem does not have a log-concave CDF in generic \footnote{In short, the random variable $w_i = \min\{v_i, A\}$ does not have a log-concave CDF in general because it has a mass point. }.

Assumption $\ref{assumption: A large}$ is used to focus our main analysis on the nondegenerate case $A \geq p_i^*, i=1,2$, where $p_i^*$ is the equilibrium price of seller $i$. This case captures the key insights of the consumer search model. When $A < p_i^*, i=1\text{ or 2}$, as previously discussed, the prominent seller effectively acts as a monopolist, requiring a new demand system to address the discontinuity between these cases; under this case, $p_i^*$ could no longer be an equilibrium price. While this discontinuity is acknowledged in the literature, it is generally not a major concern, as it does not arise in equilibrium under the special search algorithms discussed above. However, for “contracts,” this issue becomes significant because, as we will explain, “contracts” implement the largest set of equilibria. Further details are provided in Section \ref{sec: discussion}.

\section{Implementable Prices}
\label{sec:implementable prices}
Given the model, the first question we are concerned about is to what extent can a platform influence the market by committing to search algorithms. Formally, we are interested in which prices are implementable by search algorithms. The following result, analogous to the ``revelation principle'', shows that it is without loss of generality to focus on ``contracts''.

\begin{proposition}
    Any pure strategy equilibrium of the game 
can be implemented by a contract.
\end{proposition}
\begin{proof}
    Consider a search algorithm $\alpha$ and suppose $(p_1,p_2)$ is a Nash equilibrium under $\alpha$. Given $\alpha$ and $(p_1,p_2)$, there is a well-defined \textbf{corresponding contract} $\alpha^*$, expressed as $(p_1, p_2, \alpha(p_1, p_2))$. To prove the proposition, we only need to show that $(p_1,p_2)$ is a Nash equilibrium under $\alpha^*$. 
    
    Note that $0\leq\alpha\leq1$, thus by the definition of Nash equilibrium, for seller $i = 1,2$ we have: for any $p'_1$, 
    $$
    \begin{aligned}
        &\quad \alpha(p_1,p_2)\pi_1^1(p_1,p_2) + (1-\alpha(p_1,p_2))\pi_1^2(p_1,p_2) \\
        &\geq \alpha(p'_1,p_2)\pi_1^1(p_1',p_2) + (1-\alpha(p'_1,p_2))\pi_1^2(p_1',p_2) &\quad \text{(Definition of Nash Equilibrium)}\\ 
        &\geq \pi_1^2(p_1',p_2) &\quad \text{(Lemma \ref{lem: bonus})}\\ 
        &= \alpha^*(p'_1,p_2)\pi_1^1(p_1',p_2) + (1-\alpha^*(p'_1,p_2))\pi_1^2(p_1',p_2) &\quad \text{(Definition of $\alpha^*$)}
    \end{aligned}
    $$
    It means if seller 1 does not deviate from $(p_1, p_2)$ under search algorithm $\alpha$, it will not deviate from $(p_1, p_2)$ under contract $\alpha^*$. Similarly, we can show there is no profitable deviation for seller 2 under $\alpha^*$. Therefore, $(p_1,p_2)$ is a Nash equilibrium under $\alpha^*$.
\end{proof}

In words, since the profit is lower when ranked second (Lemma \ref{lem: bonus}), a corresponding contract $\alpha^*$ defined in the proof is the harshest punishment\footnote{A possibly harsher punishment is removing sellers from the platform. Then the platform can implement any prices that bring positive profits to the prominent seller \citep{BarIsaac2023Monetizing}. Our discussion about first best in Section \ref{subsec: first best} corresponds to this case. }. So if any unilateral deviation is unprofitable under a search algorithm $\alpha$, it implies any unilateral deviation is unprofitable under its corresponding contract. Hence if prices $(p_1,p_2)$ constitute a Nash equilibrium under a search algorithm $\alpha$ (no profitable unilateral deviation), then it is a Nash equilibrium under the corresponding contract $(p_1,p_2,\alpha(p_1,p_2))$.

This proposition enables us to focus on contracts without loss of generality. Since a contract can be expressed as a tuple $(p_1,p_2,\alpha)$, as we said, it can be viewed as an action recommendation from the platform. Thus, a price vector $(p_1,p_2)$ is implementable if sellers are willing to accept the recommendation. 

\begin{definition}[Implementable Prices]
    Suppose that Assumption \ref{assumption: A large} holds. Price vector $(p_1, p_2)$ is implementable if for some $\alpha \in [0,1]$, the following two IC constraints hold \footnote{Rigorously, we should express the left-hand side of the IC constraint as $\mathop{\max}_{p_i' \neq p_i} \pi_i^2(p_i', p_{-i})$ given prices $(p_i, p_j)$ since seller $i$ is only penalized when it sets a price different from $p_i$. However, we omit this notation for convenience. In fact, the IC constraints written in these two ways are equivalent. The possible discrepancy arises only when $\arg\max_{p_i'} \pi_i^2 = p_i$. In this case, we have $\max_{p_i'} \pi_i^2 = \pi_i^2(p_i, p_j)$ and $\max_{p_i' \neq p_i} \pi_i^2 \leq \pi_i^2(p_i, p_j)$, so for any $(p_i, p_j)$, the IC constraint must hold. Therefore, the two formulations do not change the IC constraints. }
\begin{equation}
    \begin{aligned}
   \text{IC1}: \mathop{\max}_{p_{1}^{\prime}} \pi_{1}^{2}(p_{1}^{\prime},p_{2})  & \leq \Pi_{1}(p_{1},p_{2})  = \pi_{1}^{2}(p_{1},p_{2})+ \alpha p_{1}B(p_1,p_2),    \\ 
   \text{IC2}: \mathop{\max}_{p_{2}^{\prime}} \pi_{2}^{2}(p_{1},p_{2}')  & \leq \Pi_{2}(p_{2},p_{1}) = \pi_{2}^{2}(p_{1},p_{2})+ (1-\alpha) p_{2}B(p_1,p_2).   \nonumber \\
   \end{aligned}
   \label{IC1&IC2}
\end{equation}
\label{def: implementable price}
\end{definition}

The IC constraints reflect the sellers’ trade-off between accepting a contract and opting for the outside option. For instance, if firm 1 rejects the contract, it can maximize \(\pi_1^2\), which could be larger than the original \(\pi_1^2\), but it forfeits the $\alpha p_1 B$ component. To ensure that sellers accept the contract, both \(\alpha p_1 B\) and \((1-\alpha) p_2 B\) must provide sufficient benefits. This implies that \(\alpha\) must be set neither too high (close to one) nor too low (close to zero).

To figure out the set of implementable prices, we define a correspondence of \textbf{implementable search orders} $\varphi: \mathbb{R}^2 \rightrightarrows [0,1]$, which maps prices to the set of feasible values of $\alpha$.
\[\varphi(p_1,p_2):=\left\{\alpha \in [0,1]: \text{IC1 } \text{ and } \text{ IC2}\right\}\]

By Definition \ref{def: implementable price}, as long as $\varphi(p_1, p_2)$ is non-empty, prices $(p_1, p_2)$ are implementable. 

\begin{corollary}
Prices $(p_1, p_2)$ are implementable if and only if $\varphi(p_1, p_2) \neq \emptyset$.
\label{cor: implementable}
\end{corollary}

With Corollary \ref{cor: implementable}, we need to figure out the behavior of $\varphi(p_1,p_2)$. Recall that the IC constraints show how the platform balances the bonus $B$ between two sellers. Since $\Pi_1(p_1,p_2)$ increases with $\alpha$, for any prices $(p_1,p_2)$, there is a lowest $\alpha$ determined by the equality of IC1, which we denote by $\underline{\alpha}$, such that IC1 holds. IC2 decreases with $\alpha$ so there is a greatest $\alpha$ determined by the equality of IC2, denoted by $\bar{\alpha}$ such that IC2 holds. Therefore, for any $(p_1,p_2)$, $\varphi(p_1,p_2) = [\underline{\alpha}, \bar{\alpha}] \cap [0,1]$. Select $(p_1,p_2)$ such that $\varphi(p_1,p_2)$ is non-empty, we obtain the set of implementable prices. 

\begin{proposition}
        Suppose that Assumption \ref{assumption: A large} holds. The set of implementable prices $P$ can be described as: 
        \begin{equation}
        P:=\left\{(p_1, p_2):\frac{\mathop{\max}\limits_{p_1^\prime} \pi_{1}^2(p_{1}^{\prime},p_{2})}{p_1}  + \frac{\mathop{\max} \limits_{p_2^\prime}\pi_{2}^2(p_{2}^{\prime},p_{1})}{p_2} + F(p_{1})F(p_{2}) \leq 1\right\}
        \end{equation}
    \label{prop: P's discrip}
\end{proposition}

\vspace{-25pt}

\begin{proof}
    As analyzed above, $$\underline{\alpha}(p_1,p_2)\triangleq\frac{\mathop{\max}\limits_{p_1^\prime} \pi_{1}^2-\pi_{1}^2}{p_1 B}\le \alpha \le 1-\frac{\mathop{\max}\limits_{p_2^\prime} \pi_{2}^2-\pi_{2}^2}{p_2 B}\triangleq \bar{\alpha}(p_1,p_2)$$

    Since $\mathop{\max}_{p_i^\prime} \pi_{i}^2\ge \pi_{i}^2$ for $i=1,2$, $\underline{\alpha}(p_1,p_2)\ge0$ and $\bar{\alpha}(p_1,p_2)\le 1$. Thus, $\varphi(p_1,p_2)\ne \emptyset$ if and only if $\underline{\alpha}(p_1,p_2)\le \bar{\alpha}(p_1,p_2)$, which can be reduced to set $P$ by noticing that $D_1^1+D^2_2=1-F(p_1)F(p_2)$.
\end{proof}

Given a price vector $(p_1,p_2)$, we interpret $\phi_i(p_i,p_{-i}) := \frac{\max_{p_i'}\pi_i^2}{p_i}$ as seller $i$'s ``virtual demand'', which is the demand needed for seller $i$ to obtain the maximum profit after deviation given prices $(p_i,p_j)$. Then Proposition \ref{prop: P's discrip} suggests that the two IC constraints can be combined into a single constraint, 
\begin{equation}
    \phi_1(p_1,p_2) + \phi_2(p_1,p_2) \leq 1 - F(p_1)F(p_2),
    \label{eq: virtual demand}
\end{equation}
which states that the total virtual demand does not exceed the total demand $1 - F(p_1)F(p_2)$. To have a further look, we state the following observation.

\begin{observation}
    For $i = 1 \text{ or 2}$, $\phi_i(p_i,p_{-i}) \leq D_i$; the equality holds if and only if $IC_i$ is binding.
    \label{cor:virtual demand}
\end{observation}

Notice that the IC constraints state for $i = 1,2$, $\max_{p_i}\pi_i^2 \leq \Pi_i$, observation $\ref{cor:virtual demand}$ is obtained immediately. When the equality holds for $i = 1 \text{ or }2$, i.e., $\max_{p_i}\pi_i^2 = \Pi_i$, we have $\phi_1(p_1,p_2) = D_1^2 + \alpha B =D_1$ or $\phi_2(p_1,p_2) = D_2^2 + (1-\alpha) B = D_2$: the virtual demand becomes the actual demand. Therefore, inequality $\ref{eq: virtual demand}$ is equivalent to requiring both IC constraints hold and the equality holds if and only if both IC constraints are binding.

With the Euclidean topology, let \(\partial P\) denote the boundary of \(P\) and \(P^{\circ}\) denote its interior. Note that \(P\) also depends on the search cost \(s\), as demand varies with \(s\). The following proposition outlines the properties of the set \(P\) and the correspondence \(\varphi\), some of which are already discussed above.

\begin{proposition}
Suppose that Assumption \ref{assumption: A large} holds. $P$ and $\varphi$ satisfy the following properties: 
    \begin{enumerate}
        \item $P$ is compact and symmetric about the line $p_1 = p_2$
        \item $P$ is convex if Assumption \ref{assumption: uniform} also holds. 
        \item For any $(p_1,p_2) \in P$, $\varphi(p_1,p_2)$ is a compact set $[\underline{\alpha}(p_1,p_2), \bar{\alpha}(p_1,p_2)]$. $\varphi(p_1,p_2)$ reduces to a singleton if and only if $(p_1,p_2) \in \partial P$. 
        \item For search costs $s_1 \leq s_2$, i.e., $A(s_1) \geq A(s_2)$, we have $P_{s_1} \subseteq P_{s_2}$. 
    \end{enumerate}
    \label{prop: P's property}
\end{proposition}

Proof is in the Appendix \ref{appendix: set P}.

When Assumption \ref{assumption: A large} holds, The above proposition characterizes the set of implementable prices $P$ and implementable search orders for each $(p_1,p_2) \in P$. Property 1 is about the geometry of set $P$. Property 3 is obtained directly from Proposition \ref{prop: P's discrip}. This is because on the boundary, two IC constraints are both required to be binding, hence $\underline{\alpha}(p_1,p_2)$ and $\bar{\alpha}(p_1,p_2)$ coincide. Regarding Property 4, the intuition is that when the search cost is low, the prominent seller’s advantage diminishes, which in turn mitigates the platform’s punishment. Thus, the platform can implement smaller sets of prices. 

Figure \ref{fig:boundary} shows an example of $P$ under uniform distribution. The blue curve represents $\partial P$. For comparison, the figure also displays equilibrium prices for prominence and random search search algorithms, i.e., $\alpha = 0,1,\frac{1}{2}$.

\begin{figure}[h]
    \centering    \includegraphics[scale=0.55]{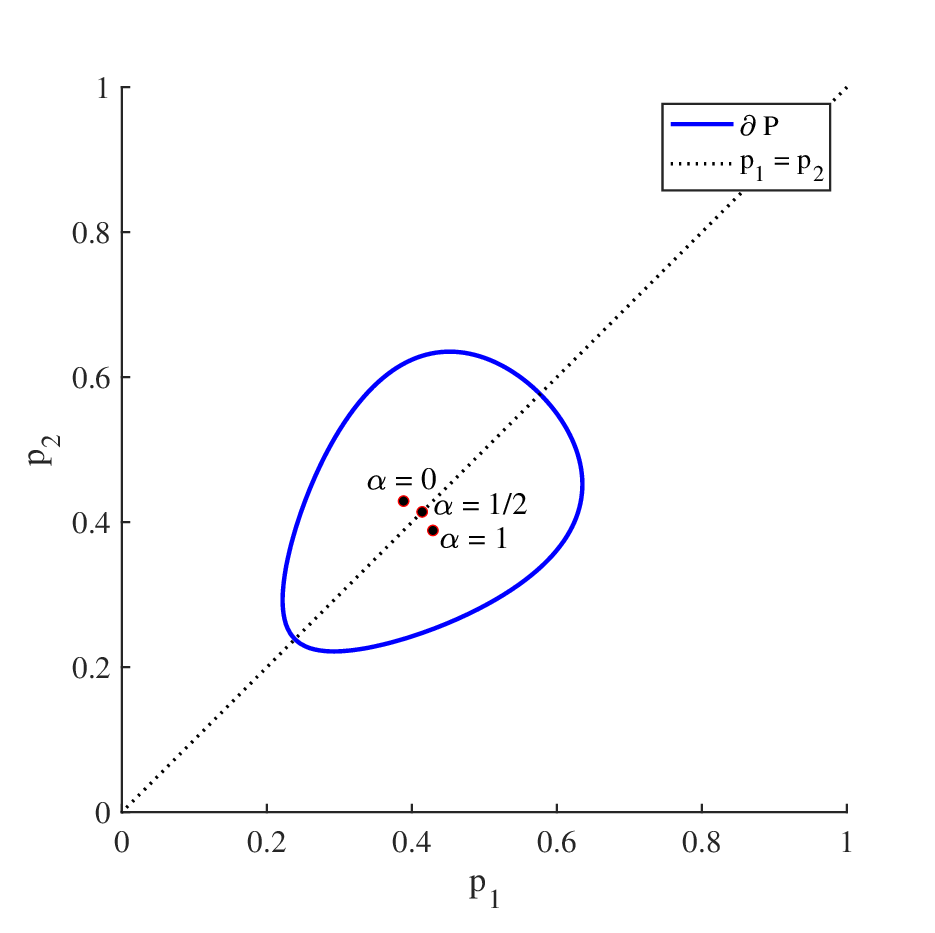}
    \caption{Boundary of $P$ and Points with Several Rank Funtions ($A = 0.7$)}
    \label{fig:boundary}
\end{figure}

We now illustrate how points in $\partial P$ are determined by the two IC constraints. Instead of considering $\alpha$ after fixing $(p_1, p_2)$ as we did previously, we now focus on the behabior of $(p_1,p_2)$ when $\alpha$ is fixed. Given $\alpha$, some $(p_1, p_2)$ are implementable. If we first consider $\alpha = 0$, the IC constraints become:
\begin{equation}
    \begin{aligned}
    \text{IC1: }&\max_{p_1^\prime}\pi_1^2(p_1^\prime,p_2)\leq \pi_1^2(p_1,p_2)\\
    \text{IC2: }& \max_{p_2^\prime} \pi_2^2(p_2^\prime,p_1) \leq \pi_2^2(p_1,p_2)+p_2 B(p_1,p_2) = \pi_2^1(p_1,p_2) 
    \end{aligned}
\end{equation}

As shown in Figure \ref{fig:formation} (a), for a given \( p_2 \), only \( p_1 = p'_1 \) satisfies IC1. Meanwhile, for a given \( p_1 \), a set of \( p_2 \)-values satisfies IC2. The intuition here is that when \( \alpha = 0 \), seller 1 will always be searched second, making it likely to deviate from any contract unless \( p_1 \) equals \( p'_1 \). However, since seller 2 is always searched first and benefits from the bonus, it is willing to accept a broader set of \( p_2 \)-values, given \( p_1 \).

As we increase \( \alpha \), seller 1 starts to receive some portion of the bonus and will, therefore, accept a wider range of \( p_1 \)-values. Conversely, as \( 1 - \alpha \) decreases with the increasing \( \alpha \), seller 2 will accept a narrower range of \( p_2 \)-values, as depicted in Figure \ref{fig:boundary} (b). When \( \alpha \) reaches 1, the path traced by the intersection of IC1 and IC2 determines the entire boundary of \( P \).

\begin{figure}[!h]
\centering
\subfigure[IC1 and IC2 when $\alpha = 0$]{
  \includegraphics[width=0.31\textwidth]{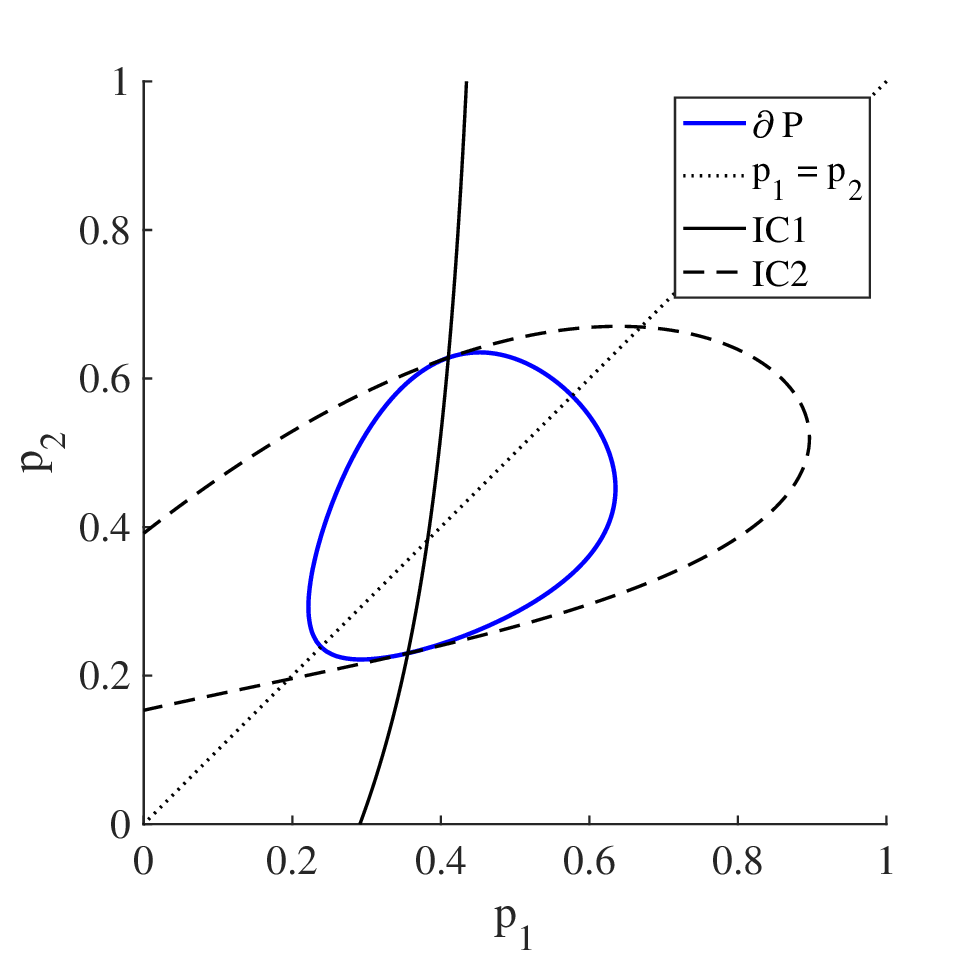}
}
\hfill
\subfigure[IC1 and IC2 when $\alpha = 0.5$]{\includegraphics[width = 0.31\textwidth]{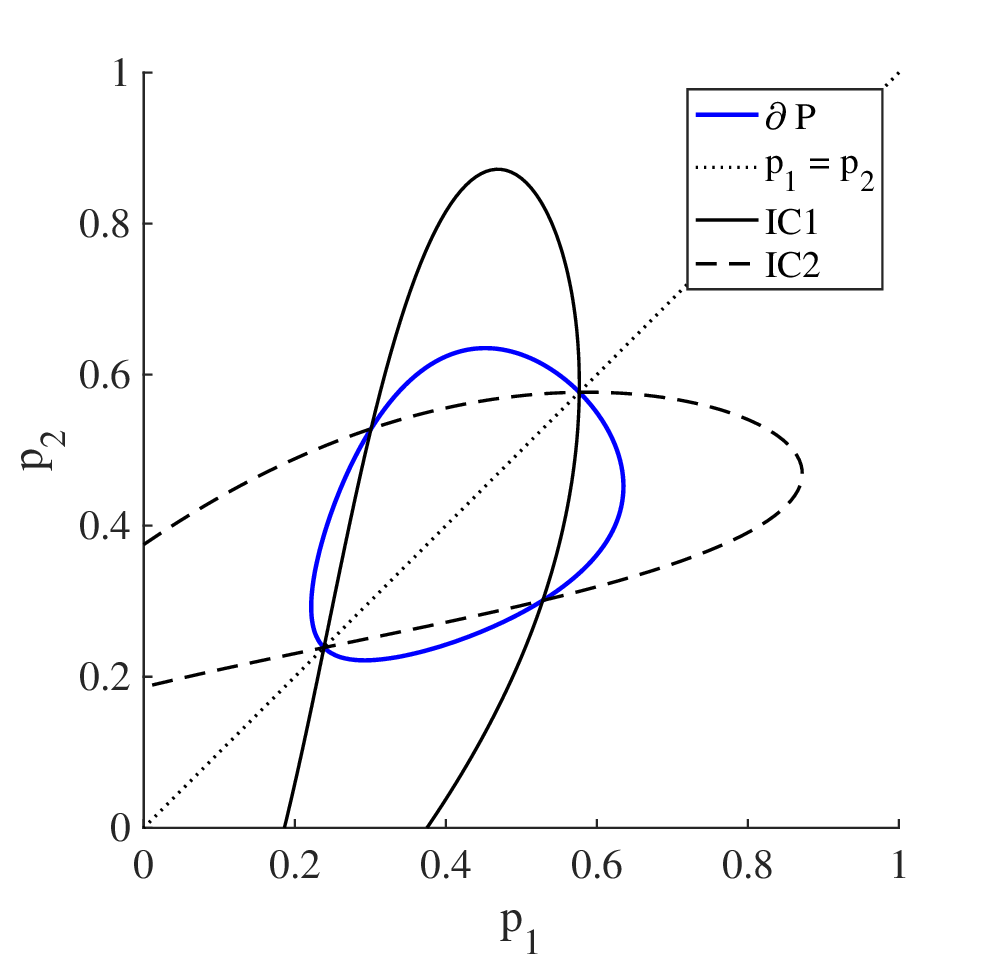}}
\hfill
\subfigure[IC1 and IC2 when $\alpha =1$]{\includegraphics[width = 0.31\textwidth]{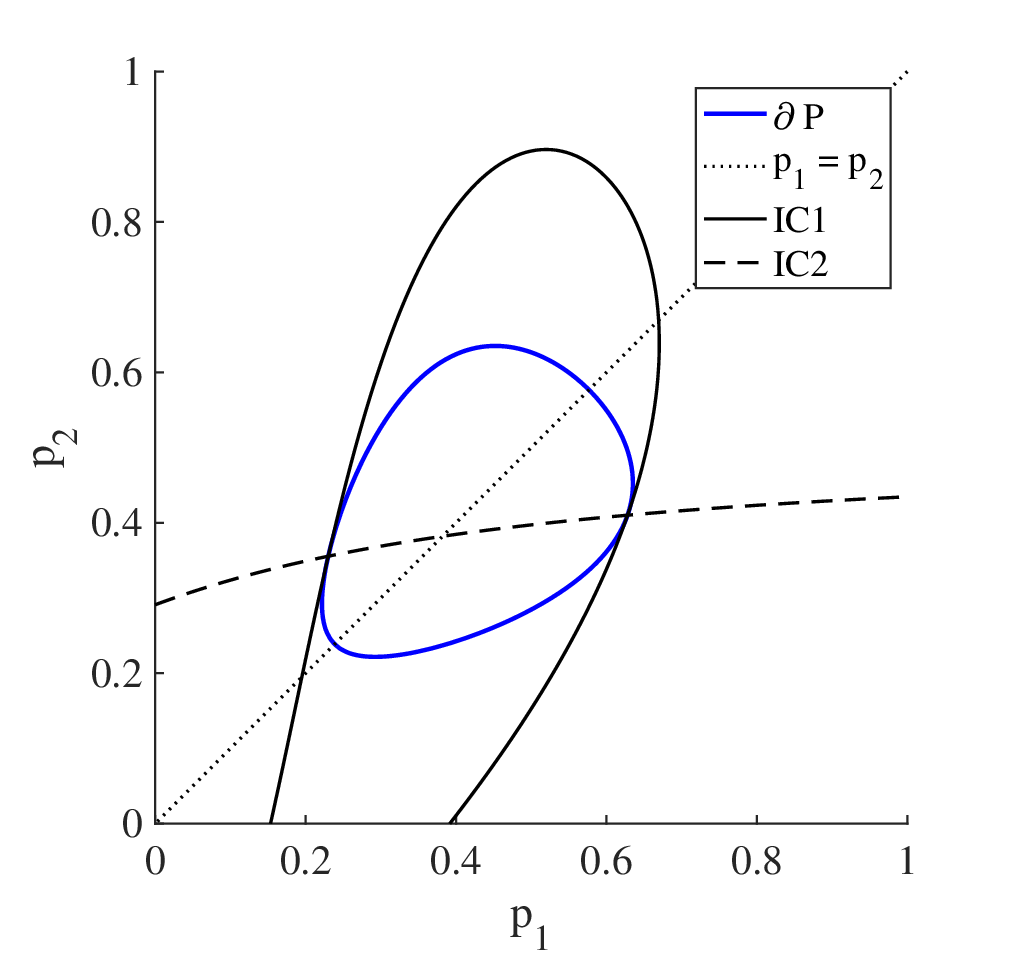}}

\caption{The Formation of set $P$ ($A = 0.7$)}
\label{fig:formation}
\end{figure}

\section{Optimal Contract}
\label{sec: optimization}

After characterizing the set of implementable prices $P$, and consequently the set of implementable contracts, the next natural question arises: which contracts are optimal given different platform objectives? Imagine that a platform charges sellers proportional commission fees, then the objective is to maximize the industry profit. Common objectives besides industry profit include total demand, consumer surplus and social welfare.   

We first formalize the optimization problems. The platform has three decision variables: $p_1$, $p_2$, and $\alpha$. An objective function is thus defined as a function of $p_1, p_2$ and $\alpha$. Specifically, we denote the objective function by \( J \).
$$
J: \mathbb{R}^2 \times [0,1] \to \mathbb{R}
$$

The platform's optimization problem is 
\[
\begin{aligned}
    &\mathop{\max(\min)}_{p_1,p_2,\alpha} \quad J(p_1, p_2, \alpha)\\
    & \text{subject to} \quad (p_1, p_2) \in P,\ \ \alpha \in \varphi(p_1, p_2)
\end{aligned}
\]

We simplify the optimization problem by first considering how a platform determines $\alpha$ given prices $(p_1,p_2)$, i.e., given different objective functions, what is the optimal way for a platform to allocate its traffic to sellers. Formally, given any $(p_1,p_2)\in P$, we want to find the optimal rule for choosing the search algorithm $\hat{\alpha}(p_1,p_2):P\to \varphi(p_1,p_2)$ to get the maximized value function $\bar{\mu}(p_1,p_2)\triangleq \max_{\alpha\in \varphi(p_1,p_2)} J(p_1,p_2,\alpha)=J(p_1,p_2,\hat{\alpha}(p_1,p_2))$ or the minimized value function $\underline{\mu}(p_1,p_2)\triangleq \min_{\alpha\in \varphi(p_1,p_2)} J(p_1,p_2,\alpha)=J(p_1,p_2,\hat{\alpha}(p_1,p_2))$. 

The above steps reduce the optimization problem to the space of $P$. Next, we study how should a platform determine the prices of contracts. Formally, an optimization problem becomes: 
$\max_{(p_1,p_2)\in P} \bar{\mu}(p_1,p_2)$ or $\min_{(p_1,p_2)\in P}\underline{\mu}(p_1,p_2)$.

We use \(J^{(\cdot)}\) and \(\mu^{(\cdot)}\) with different superscripts to denote various objective functions and their corresponding optimal value functions. We use \(\Pi\), \(SW\), \(TP\), and \(CS\) to represent total industry profit, social welfare, trade probability, and consumer surplus, respectively. Depending on \(J\), we state the following lemma:

\begin{lemma}[Optimal Traffic Allocation]
For $(p_1,p_2)\in P$, \\
When $p_1 > p_2$, 
    \begin{itemize}
        \item $\bar{\mu}^{\Pi}(p_1,p_2) = J^{\Pi}(p_1, p_2, \bar{\alpha}), \underline{\mu}^{\Pi}(p_1,p_2) = J^{\Pi}(p_1, p_2, \underline{\alpha} )$
        \item $\bar{\mu}^{SW}(p_1,p_2) = J^{SW}(p_1, p_2, \underline{\alpha} ), \underline{\mu}^{SW}(p_1,p_2) = J^{SW}(p_1, p_2, \bar{\alpha} )$
    
    \end{itemize}
When $p_1 < p_2$
    \begin{itemize}
        \item $\bar{\mu}^{\Pi}(p_1,p_2) = J^{\Pi}(p_1, p_2, \underline{\alpha} ), \underline{\mu}^{\Pi}(p_1,p_2) = J^{\Pi}(p_1, p_2, \bar{\alpha} )$
        \item $\bar{\mu}^{SW}(p_1,p_2) = J^{SW}(p_1, p_2, \bar{\alpha} ), \underline{\mu}^{SW}(p_1,p_2) = J^{SW}(p_1, p_2, \underline{\alpha} )$
    \end{itemize}
When $p_1 = p_2$, for all $\alpha \in [0,1]$, 
\begin{itemize}
    \item $\bar{\mu}^{\Pi}(p_1,p_2) = \underline{\mu}^{\Pi}(p_1,p_2) = J^{\Pi}(p_1, p_2, \alpha )$
    \item $\bar{\mu}^{SW}(p_1,p_2) = \underline{\mu}^{SW}(p_1,p_2) = J^{SW}(p_1, p_2, \alpha )$
\end{itemize}
\label{lemma: pin down alpha}
\end{lemma}

The proof is in Appendix \ref{appendix: proof of main results}

In the above results about $J^{\Pi}$, Lemma \ref{lem: bonus} again plays an important role: notice that the original form of $J^{\Pi}$ is
$J^{\Pi} =  \pi_{1}^{2}+\pi_{2}^{2}+ \alpha p_{1}B_1+(1-\alpha) p_{2}B_2$, and we have $\frac{\partial J^{\Pi}}{\partial \alpha} = p_1B_1 - p_2B_2$. According to Lemma $\ref{lem: bonus}$'s result about $B = B_1 = B_2$, $\frac{\partial J^{\Pi}}{\partial \alpha} = (p_1-p_2)B$. The non-negativity of $B$ ensures that the sign of $\frac{\partial J^{\Pi}}{\partial \alpha}$ only depends on the sign of $p_1 -  p_2$. The above results about $J^{\Pi}$ thus follows. Intuitively, given the products' prices, to maximize the total profit, the platform wants to allocate more traffic to the seller with the higher price, i.e., to make the expensive product more likely to be prominent. 

In terms of social welfare $J^{SW}$, we show in the Appendix \ref{appendix: set P} that the sign of $\frac{\partial J^{SW}}{\partial \alpha}$ is the inverse of the sign of $p_1 - p_2$. It means there are conflicts between social welfare and sellers' welfare. To maximize social welfare, the platform should instead allocate more traffic to the cheaper product. Intuitively, this traffic allocation policy benefits consumers since it is consistent with the optimal selection rule, which says the consumer should search from cheap products to expensive products, despite that $\alpha$ might not be able to take $0$ or $1$.

Finally, when $p_1 = p_2$, both $J^{\Pi}$ and $J^{SW}$ do not depend on $\alpha$. To see why, consider social welfare as an example. Use $SW_1(p_1,p_2)$ and $SW_2(p_1,p_2)$ to denote the social welfare given $(p_1,p_2)$ when seller $1$ and $2$ is searched first, respectively. Thus, $J^{SW} = \alpha SW_1 + (1-\alpha) SW_2$. Since $\alpha$ is just a weight for $SW_1$ and $SW_2$, when $p_1 = p_2$, $SW(p) = SW_1(p) = SW_2(p)$ since $SW_2(p_1,p_2) = SW_1(p_2,p_1)$. Thus, $J^{SW} = SW(p)$ which is irrelevant to $\alpha$. The logic is the same for $J^{\Pi}$. It means when sellers' prices are the same, the platform does not care how the traffic is allocated since whatever $\alpha$ does not change the sum of social welfare or industry profit. 

\subsection{Seller Optimal Contract: First Best}\label{subsec: first best}

Among the possible objectives of a platform, we first consider how a platform maximizes $J^{\Pi} = \Pi_1 + \Pi_2$ subject $(p_1,p_2) \in P$ and $\alpha \in \varphi(p_1,p_2)$. Notably, by removing the constraints we obtain the first best for this problem, and the platform effectively becomes a multiproduct seller. How should this multiproduct seller determine prices and search order in a sequential search environment to maximize its total profit is an open question in the literature \citep[page 1275]{choi_consumer_2018}. Our results in the following analysis suggest that for ex-ante symmetric sellers, ordered search (where $\alpha = 0$ or $\alpha = 1$) is optimal for the multiproduct seller.

Formally, the multiproduct seller's problem is 
\[
\begin{aligned}
   &\max_{p_1,p_2,\alpha}J^{\Pi} =  \pi_{1}^{2}+\pi_{2}^{2}+ \alpha p_{1}B+(1-\alpha) p_{2}B\\
   &s.t. \quad 0 \leq \alpha \leq 1
\end{aligned}
\]

The following proposition says that (at least) under uniform distributions, ordered search is optimal for a multiproduct seller though products are ex-ante symmetric for consumers.

\begin{proposition}[First Best]
       Suppose that Assumption \ref{assumption: uniform} holds. For a multiproduct seller, the optimal contract takes the form $(p_1^*, p_2^*, 1)$ or $(p_2^*, p_1^*, 0)$, where \(p_1^* > p_2^*\).
       \label{prop:first best}
\end{proposition}

From the analysis of Lemma \ref{lemma: pin down alpha}, when $p_1 > p_2$, the multiproduct seller should allocate more traffic to seller 1, i.e., to choose $\alpha = 1$. When $p_1< p_2$, it is optimal to choose $\alpha = 0$ and when $p_1 = p_2$, the total profit is irrelevant to $\alpha$. In the above three cases, we can calculate the optimal prices and the corresponding total profits. To show that ordered search is optimal, It then suffices to show the profit when $p_1 = p_2$ is lower than the profit when $p_1 > p_2$ or $p_1 < p_2$. When $p_1 = p_2$, we denote the optimal price as $p^*$. Under the uniform distribution, $p^* = \frac{\sqrt{3}}{3}$. Since in this case, the total profit does not depend on $\alpha$, we can hypothetically set $\alpha = 1$ and will find that $p^*$ does not satisfy the first order condition for optimality. Details can be found in Appendix \ref{appendix: proof of main results}. 

Despite the asymmetric structure of the optimal contract, there is a connection between the multiproduct seller’s problem and the platform’s problem. Specifically, $p^*$ serves as a critical value in determining the symmetry of optimal contracts for the platform's problem.  

\subsection{Seller Optimal Contract}
Now we formally consider the platform's problem of maximizing the industry profit. This optimization problem yields the seller-optimal contract. We define the hihest and lowest point on the diagonal of $\partial P$ by $(\bar{p}, \bar{p}), (\underline{p}, \underline{p})$,
$$
\{(\bar{p}, \bar{p}), (\underline{p}, \underline{p})\} = \{(p_1, p_2) \in \partial P : p_1 = p_2\}, \bar{p} \geq \underline{p}
$$

\begin{theorem}
    Suppose that Assumption \ref{assumption: uniform} and \ref{assumption: A large} hold. The seller-optimal contract is $(\bar{p},\bar{p},1/2)$ when $\bar{p} \leq p^* $. When $\bar{p} > p^* $, the seller-optimal contract is asymmetric ($p_1\ne p_2$ and $\alpha \ne 1/2$). The first best is unattainable under both circumstances.
    \label{thm1}
\end{theorem}

Recall from the last property from Proposition \ref{prop: P's property}, the set \( P \) expands as the search cost \( s \) increases. Thus, the condition \( p^* = \bar{p} \) identifies a critical value of search cost \( s^* \). Intuitively, this condition determines whether the optimization problem has an interior solution. When the set is small so that the interior is not included, the best boundary point is the highest point on the diagonal. When the interior is in $P$, we can use the same logic as the first best to show it is asymmetric. However, our result is still non-trivial because $p^*$ is not the optimal solution for the first best. In other words, we cannot directly conclude that the interior solution lies in $P$ based on $p^* < p$ alone. We will see this through the sketch of the proof.

Specifically, we prove Theorem \ref{thm1} in the following steps:
first we use Lemma $\ref{lemma: pin down alpha}$ to reduce the optimization problem to the set $P$. This step simplifies the analysis but retains complexity since now $\alpha$ becomes a function of prices. However, notice that $\underline{\alpha}$ and $\bar{\alpha}$ are determined by the equality of IC1 and IC2, respectively. Thus, depending on which IC constraint is binding, we have 
$$\max_{p_1'}\pi_1^2(p_1',p_2) = \Pi_1 \text{ or } \max_{p_2'}\pi_2^2(p_1,p_2') = \Pi_2
.$$
More importantly, on the boundary of $P$, both equality holds. 

Now as an example, consider the case when $p_1 \leq p_2$. We have $\bar{\mu}(p_1,p_2) = J(p_1,p_2,\underline{\alpha})$, which implies $\max_{p_1'}\pi_1^2(p_1',p_2) = \Pi_1$. Thus, 
$$
\bar{\mu}^\Pi = \max_{p_1'}\pi_1^2(p_1',p_2) + \Pi_2(p_1,p_2,\underline{\alpha}).
$$

Notice that \(\max_{p_1'} \pi_1^2(p_1', p_2)\) does not depend on \( p_1 \). This observation enables us to show that when \(\bar{p} \leq p^*\), any point in the interior of \( P^\circ \) can be shifted horizontally to either the boundary \(\partial P\) or the diagonal to achieve higher total profit. This step reduces the optimization problem to the boundary \(\partial P\) and the diagonal. The condition \(\bar{p} \leq p^*\) is critical here because \(\frac{\partial \Pi_2(p_1, p_2, \underline{\alpha})}{\partial p_1}\) is maximized at \((\bar{p}, \bar{p})\), and \( p^* \) is the point where this derivative equals zero. Thus, \( p^* \) determines whether the interior solution is included in \( P \), even though \( p^* \) itself is not the optimal solution, as shown in Proposition \ref{prop:first best}. On the boundary, where both IC constraints bind, the objective function can be analyzed along the curve \(\partial P\). This analysis confirms that the optimal solution is \((\bar{p}, \bar{p})\). Detailed derivations and proofs are provided in the Appendix \ref{appendix: proof of main results}.

For the impossibility result concerning the first best, we observe that the first-order condition for achieving the first best requires \(\frac{\partial (\pi_{2}^{1} + \pi_{1}^{2})}{\partial p_{1}} = 0\) and \(\alpha = 0 \ \ \text{or} \ \ 1\). However, when \(\alpha = 0 \ \ \text{or} \ \ 1\), one of the IC constraints is binding. Suppose \(\alpha = 0\), then IC1 is binding, meaning \(\frac{\partial \pi_{1}^{2}}{\partial p_{1}} = 0\). Since \(\frac{\partial \pi_{2}^{1}}{\partial p_{1}} > 0\), the first-order condition for the first best is not satisfied. The same logic applies when \(\alpha = 1\). Therefore, the first best for a multiproduct setting is unattainable through contract design.

\subsection{Trade Probability and Welfare}

Another common business model for platforms is the volume-based commission fee, where revenue depends on total demand. Under this model, the platform may seek to maximize total demand. The following result shows that, in contrast to maximizing industry profit, total demand is maximized at the lowest symmetric point of $P$, namely, $(\underline{p}, \underline{p})$.

\begin{proposition}
    Suppose that Assumption \ref{assumption: uniform} holds. The trade probability is maximized at contract $(\underline{p}, \underline{p},\frac{1}{2})$. 
    \label{prop:trade probability}
\end{proposition}

Intuitively, lower prices lead to higher demand. When Assumption \ref{assumption: uniform} holds, the objective function is given by
\[
J^{TP} = 1 - p_1 p_2,
\]
which represents the number of consumers, subtracting those who opt for the outside option. However, the random search structure still needs proof. To establish this result, it is necessary to analyze the objective function along the curve $\partial P$, which is non-trivial. Details can also be found in the Appendix \ref{appendix: proof of main results}.

The same contract also maximizes social welfare and consumer surplus while minimizing industry profit. 

\begin{proposition}
    Suppose that Assumption \ref{assumption: uniform} and \ref{assumption: A large} hold. At contract $(\underline{p}, \underline{p}, \frac{1}{2})$, the industry profit is minimized, while social welfare and consumer surplus are maximized. 
    \label{prop: welfare max}
\end{proposition}

Proposition \ref{prop: welfare max} reveals that the lowest symmetric prices combined with random search are socially optimal but least favorable for sellers or a platform operating under a proportional commission fee model. This highlights a fundamental misalignment between the incentives of a platform using a proportional fee model and societal welfare. In contrast, a platform that generates revenue based on demand rather than profit aligns its incentives with those of society.

It is worth noting that proving Proposition \ref{prop: welfare max} is challenging, as the functional form of social welfare is more complex, preventing the use of IC constraints to simplify the problem. However, the result can still be established using monotonicity arguments, which allow us to restrict our analysis from the entire domain to the boundary and ultimately to an arc near the lowest price point. Details of the proof can be found in Appendix \ref{appendix: proof of main results}.

\section{Discussions}\label{sec: discussion}
\subsection{Corner Solution}
According to Proposition \ref{prop: P's property}, \( P \) becomes larger as \( A \) decreases. However, this expanded \( P \) may include prices higher than \( A \), leading to two issues. Consider a point \((p_1, p_2) \in P\) where \( p_1 \geq A \) and \( p_2 < A \) (the case where \( p_1 < A \) and \( p_2 \geq A \) is analogous). 

First, when seller 1 is ranked second—which generically occurs since \((p_1, p_2) \in P\)—a consumer will search seller 1 only if 
\[
u_2 - p_2 < A - p_1 \leq 0.
\]
This condition implies that the consumer should instead take the outside option, meaning they will never search seller 1. As a result, seller 1’s demand when ranked second is zero, i.e., \( D_1^2 = 0 \).

Second, if \( u_2 \geq p_2 \), the consumer will immediately buy from seller 2, resulting in 
\[
D_2^1 = 1 - F(p_2),
\]
rather than the previously derived \( D_2^1 \). Consequently, the same IC1 and IC2 conditions cannot be used to delineate the region.

As illustrated in Figures \ref{Corner}(a) and \ref{Corner}(b), the valid region for \( P \) lies strictly below the lines \( p_1 = A \) and \( p_2 = A \). The distinction between these figures lies in whether \( P \) includes points where both \( p_1 \geq A \) and \( p_2 \geq A \). The monotonicity of \( P \) with respect to \( A \) indicates that such points can exist when \( A \) is sufficiently small.

\begin{figure}[h]
    \centering
\subfigure[A = 0.5]{\includegraphics[width = 0.32\textwidth]{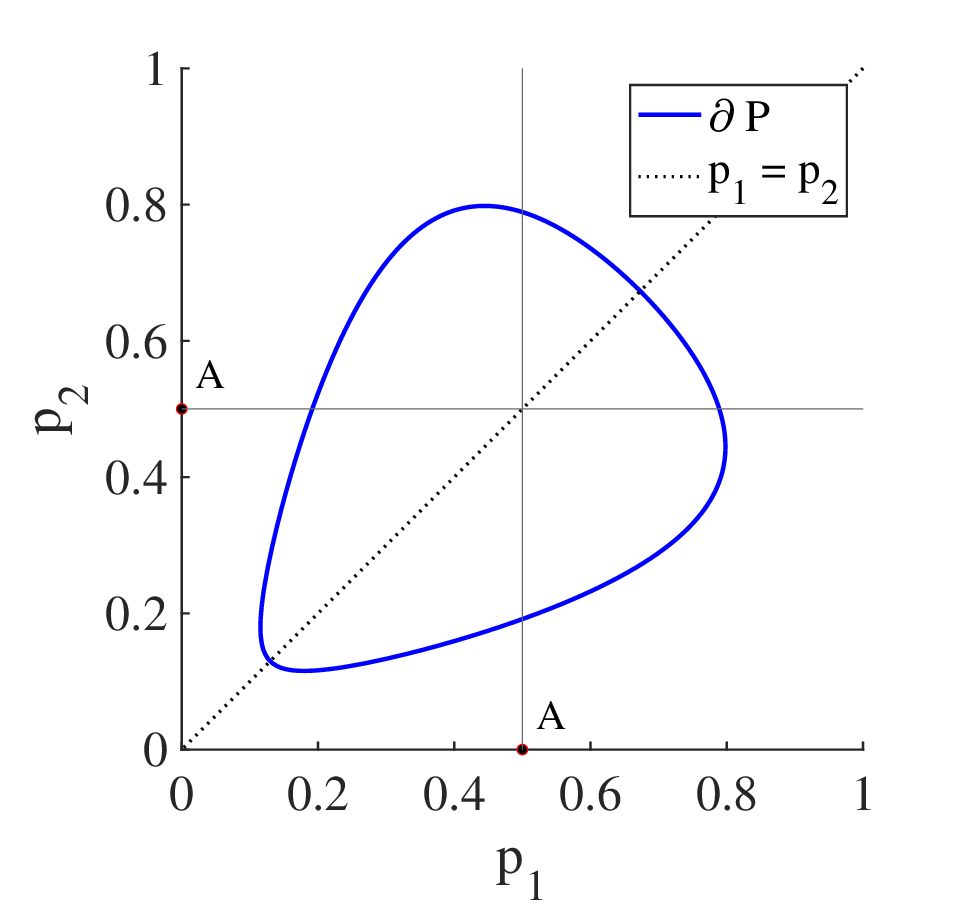}}
\hfill
\subfigure[A = 0.65]{\includegraphics[width = 0.32\textwidth]{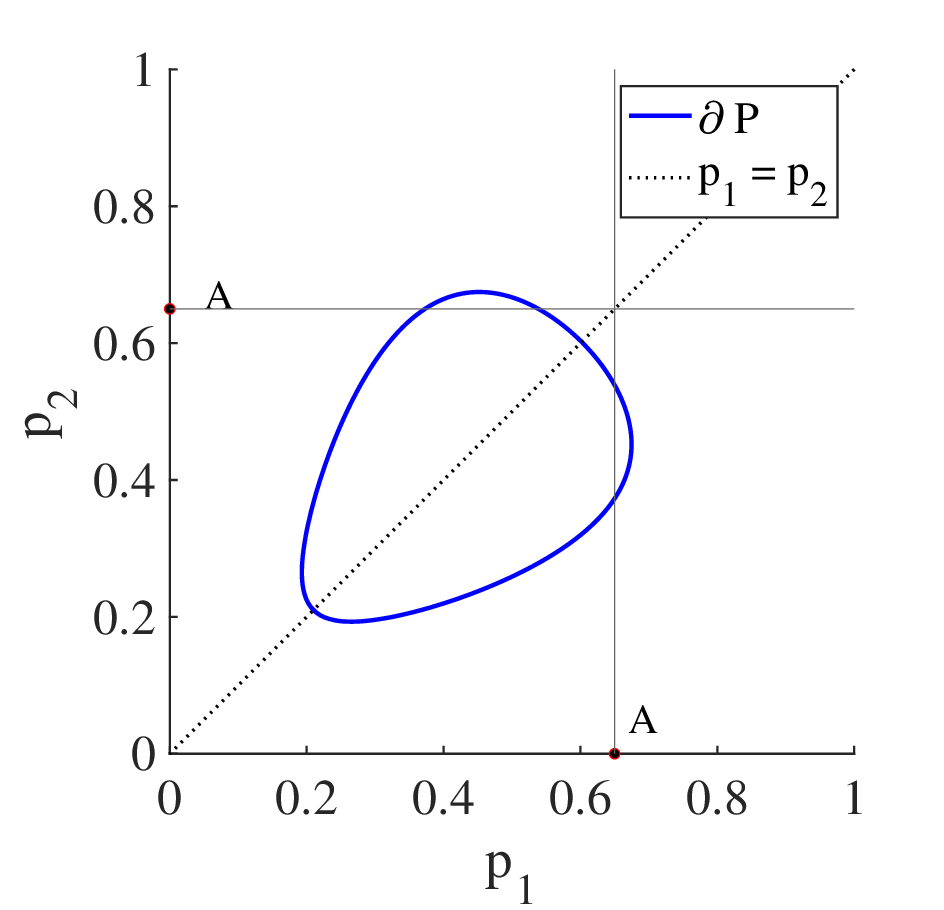}}
\hfill
\subfigure[A = 0.7]{\includegraphics[width = 0.32\textwidth]{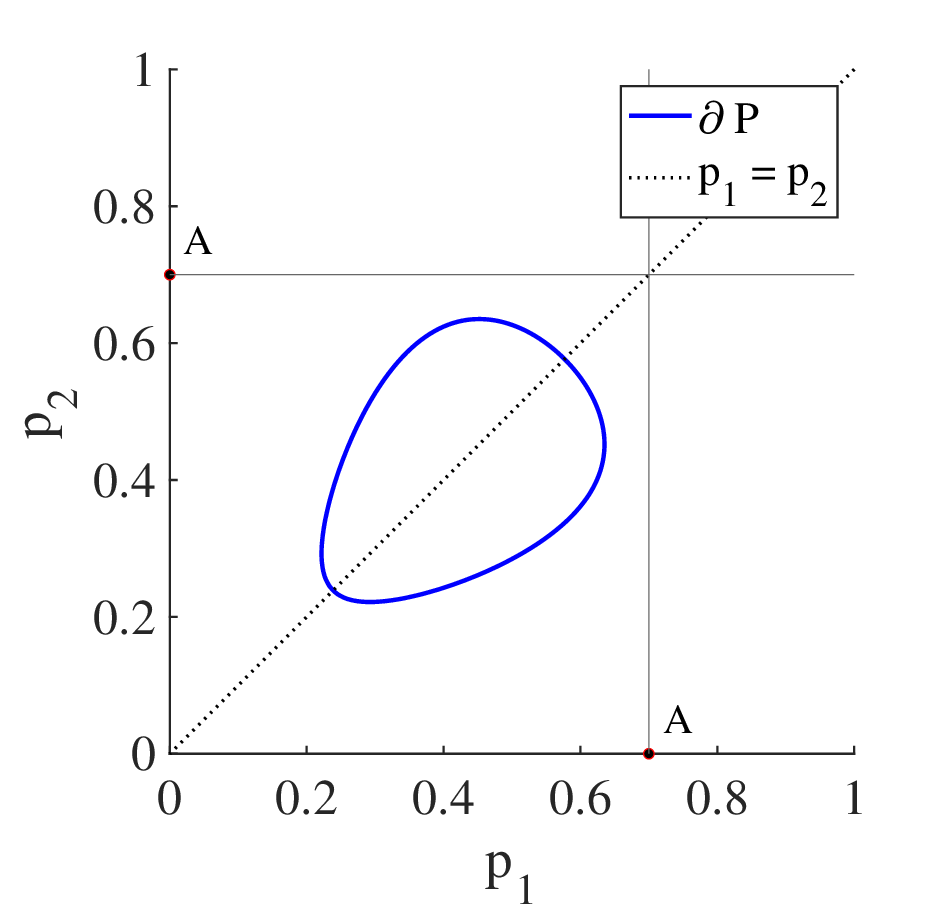}}
\caption{Three Possible Relative Positions for $P$ and $p = A$}
\label{Corner}
\end{figure}

However, some prices larger than $A$ are still implementable with modified IC constraints. Given prices $(p_1,p_2)$, first consider $p_1 \geq A, p_2<A$. Thus, $(p_1,p_2)$ is implementable if for some $\alpha \in [0,1]$ the following IC constraints hold. 

\[
\begin{aligned}
 &\widehat{\text{IC1}}: \max_{p_1'}\pi_1^2(p_1',p_{2}) \leq \alpha \pi_1^1(p_1,p_{2}) \\
 & \widehat{\text{IC2}}: \mathop{\max}_{p_{2}^{\prime}} \pi_{2}^{2}(p_{1},p_{2}')   \leq (1-\alpha) \hat{\pi}_2^1(p_1,p_2) + \alpha \pi_2^2(p_1,p_2) 
\end{aligned}
\]

The modified IC constraints address the problems mentioned above: in $\widehat{\text{IC1}}$, seller 1's profit when ranked second is deleted since the demand $D_1^2 = 0$ and in $\widehat{\text{IC2}}$, we use $\hat{\pi}_2^1:= p_2(1-F(p_2))$ to denote the profit of seller 2 when ranked first. Similarly, when $p_2\geq A$ and $p_1 < A$, define $\hat{\pi}_1^1:=p_1(1-F(p_1))$. The IC constraints are:
\[\widehat{\text{IC1}'}: \mathop{\max}_{p_{1}^{\prime}} \pi_{1}^{2}(p_{1}^{\prime},p_{2})   \leq \alpha \hat{\pi}_1^1 + (1-\alpha)\pi_1^2; \quad \widehat{\text{IC2}'}: \max_{p_2'}\pi_2^2(p_1,p_{2}') \leq (1-\alpha) \pi_2^1.\]

When $p_1,p_2\geq A$, the IC constraints are:
\[\widehat{\text{IC1}''}: \mathop{\max}_{p_{1}^{\prime}} \pi_{1}^{2}(p_{1}^{\prime},p_{2})   \leq \alpha \hat{\pi}_1^1 ; \quad \widehat{\text{IC2}''}: \max_{p_2'}\pi_2^2(p_1,p_{2}') \leq (1-\alpha) \hat{\pi}_2^1.\]

To facilitate the comparison, we suppose $P$ includes the part above $p_1 = A$ and $p_2 = A$. We denote the set of implementable prices under those modified IC constraints by $\hat{P}$ and state the following observations.

\begin{observation}
Suppose that $P$ includes the area when $p_1 \geq A$ or $p_2 \geq A$. We have: 
\begin{itemize}
    \item $(p_1,p_2) \in P$ $(p_2\geq A, p_2<A)$ determined by $\alpha = 0$ satisfies $\widehat{\text{IC1}'}$ and $\widehat{\text{IC2}'}$. $(p_1,p_2) \in P$ $(p_1\geq A, p_2<A)$ determined by $\alpha = 1$ satisfies $\widehat{\text{IC1}}$ and $\widehat{\text{IC2}}$. 
    \item $\hat{P}  \subseteq P$. 
\end{itemize}
\label{observation:corner}
\end{observation}

The proof is in the Appendix \ref{appendix: proof of main results}. Simply speaking, though some points are still implementable when at least one seller's price is above $A$, the set of implementable prices is smaller compared to the untruncated $P$.

\newpage

\appendix

\section{The Sequential Search Problem}
\label{appendix: search problem}

\subsection{Consumer search problem}

With the sequential search framework, we can calculate the demand for each seller ranked first and second. Let \(\Delta p=p_{i}-p_{j}\). When seller \(i\) is searched first and seller $j$ is searched second,

\[\begin{aligned}
    D_{i}^{1}&= 1-F(A+\Delta p)+\int_{p_{i}}^{A+\Delta p}F(u-\Delta p)f(u)du \\
    &= 1-F(A+\Delta p) + \int_{p_{j}}^{A}F(u)f(u+\Delta p)du\\
    D_{j}^{2}&=(1-F(A))F(A+\Delta p)+ \int_{p_{j}}^{A} F(u+\Delta p)f(u)du\\
    &=(1-F(A))F(A+\Delta p)+\int_{p_{i}}^{A+\Delta p}F(u)f(u-\Delta p)du
\end{aligned}\]
where \(p_j \leq A\), and \(A+\Delta p \leq 1\). 

Similarly, when consumers search \(j\) first, we have
\[\begin{aligned}
D_{j}^{1}&=1-F(A-\Delta p)+ \int_{p_{i}}^{A} F(u)f(u-\Delta p)du \\
D_{i}^{2}&=(1-F(A))F(A-\Delta p)+ \int_{p_{j}}^{A-\Delta p}F(u)f(u+\Delta p)du.
\end{aligned}\]
where \(p_i \leq A\), and \(A-\Delta p \leq 1\). 

By the definition of the bonus $B_i(p_1,p_2)$:
\[B_{i}= 1-F(A+\Delta p) -(1-F(A))F(A-\Delta p)+\int_{A-\Delta p}^{A}F(u)f(u+\Delta p)du\]
\[B_{j}=1-F(A-\Delta p) -(1-F(A))F(A+\Delta p)+\int_{A+\Delta p}^{A}F(u)f(u-\Delta p)du\]

Last, when seller \(i\) is searched first and seller $j$ is searched second, the social values from buying are
\[V_i^1=\int^1_{A+\Delta p}uf(u)du+\int^{A+\Delta p}_{p_i}uF(u-\Delta p)f(u)du\]
\[V_j^2=F(A+\Delta p)\int^1_{A}uf(u)du+\int^{A}_{p_j}uF(u+\Delta p)f(u)du\]

\subsection{Lemmas on Consumer Search Problem} 

\begin{lemma}
    (Properties of Bonus) Given prices $p_1, p_2$,
    \begin{enumerate}
        \item \(B_{i}\)= \(B_{j}\).
        \item \(B_{i}\geq 0 \), the equality holds if and only if  \(s=0\).
        \item \(B_{i}\) strictly increases with \(s\).
    \end{enumerate}
\end{lemma}

\begin{proof}  
\( D_i^1+D_j^2=D_i^2+D_j^1=1-F(p_i)F(p_j)\), thus \(B_{i} =D_i^1-D_i^2=D_j^1-D_j^2= B_{j}
    \).
\[
\begin{aligned}
       B_{i} &= 1-F(A+\Delta p) -(1-F(A))F(A-\Delta p)+\int_{A-\Delta p}^{A}F(u)f(u+\Delta p)du \\
       &=1-F(A+\Delta p)-(1-F(A+\Delta p))F(A-\Delta p)\\
       &-(F(A+\Delta p)-F(A))F(A-\Delta P) + \int_{A-\Delta p}^{A}F(u)f(u+\Delta p)du\\
       &=(1-F(A+\Delta p))(1-F(A-\Delta p))+\int_{A-\Delta p}^{A}(F(u)-F(A-\Delta p))f(u+\Delta p)du \geq 0
\end{aligned}
\]  

\(B_{i} = 0\) if and only if \(A=1\), that is, \(s=0\).

\(\frac{\partial B_{i}}{\partial A}=-[1-F(A)][f(A+\Delta p)+f(A-\Delta p)] \leq 0\). \(B_{i}\) strictly decreases with \(A\) and thus strictly increases with \(s\).
\end{proof}

\begin{lemma}
    (Properties of Demands and Profits)
    \begin{enumerate}
        \item \(\frac{\partial D_{i}^{1}}{\partial p_{j}}=\frac{\partial D_{j}^{2}}{\partial p_{i}} \geq 0\), the equality holds if and only if \(A=p_j=1\).
        \item \(\frac{\partial D_{i}^{1}}{\partial p_{i}}<0 ,\frac{\partial D_{j}^{2}}{\partial p_{j}} \leq 0\), the equality holds if and only if \(A=p_j=1\) and \(p_i=0\).
        \item  \(\frac{\partial^2 D_{i}^{1}}{{\partial p_{i}}^2} \leq 0 \) when \(F\) is twice differentiable and convex.
        
    \end{enumerate}
        \label{BestResponse}
\end{lemma}

\begin{proof}
    \[
    \begin{aligned}
            \frac{\partial D_{i}^{1}}{\partial p_{j}} &= f(A+\Delta p)-F(A)f(A+\Delta p)+\int_{p_{i}}^{A+\Delta p}f(u-\Delta p)f(u)du \\
            &=(1-F(A))f(A+\Delta p)+ \int_{p_{j}}^{A}f(u+\Delta p)f(u)du =\frac{\partial D_{j}^{2}}{\partial p_{i}} \geq 0,          
    \end{aligned}
  \]
the inequality becomes equality if and only if \(A=p_j=1\).

  \[
    \begin{aligned}
            \frac{\partial D_{i}^{1}}{\partial p_{i}} &= -f(A+\Delta p)+\int_{p_{j}}^{A}F(u)f^{\prime}(u+\Delta p)du\\
            &=-(1-F(A))f(A+\Delta p)-F(p_j)f(p_i)-\int_{p_{j}}^{A}f(u)f(u+\Delta p)du\\
            &=-\frac{\partial D_{i}^{1}}{\partial p_{j}}-F(p_j)f(p_i)  < 0,
    \end{aligned}
  \]
the equality doesn't hold since \(F(p_j)=1\) when \(p_j=1\).
    \[
    \begin{aligned}
            \frac{\partial D_{j}^{2}}{\partial p_{j}} &= -(1-F(A))f(A+\Delta p)-F(p_i)f(p_j)-\int_{p_{j}}^{A}f(u+\Delta p)f(u)du\\
            &=-\frac{\partial D_{i}^{1}}{\partial p_{j}}-F(p_i)f(p_j)  \leq 0,
    \end{aligned}
  \]
the inequality becomes equality if and only if \(A=p_j=1\) and \(p_i=0\).

  When \(F\) is twice differentiable and convex, we have:
  \[
    \begin{aligned}
            \frac{\partial^2 D_{i}^{1}}{{\partial p_{i}}^2}
            &=-(1-F(A))f^{\prime}(A+\Delta p)-F(p_j)f^{\prime}(p_i)-\int_{p_{j}}^{A}f(u)f^{\prime}(u+\Delta p)du \leq 0
    \end{aligned}
  \]
  \[
    \begin{aligned}
            \frac{\partial^2 D_{i}^{1}}{\partial p_i \partial p_j} = \frac{\partial^2 D_{j}^{2}}{{\partial p_i}^2} 
            &=(1-F(A))f^{\prime}(A+\Delta p)+\int_{p_{j}}^{A}f(u)f^{\prime}(u+\Delta p)du \geq 0
    \end{aligned}
  \]
 \[
    \begin{aligned}
            \frac{\partial^2 D_{i}^{1}}{{\partial p_j}^2} = \frac{\partial^2 D_{j}^{2}}{\partial p_i \partial p_j} 
            &=-(1-F(A))f^{\prime}(A+\Delta p)-f(p_i)f(p_j)-\int_{p_{j}}^{A}f(u)f^{\prime}(u+\Delta p)du \leq 0
    \end{aligned}
  \]
\[
    \begin{aligned}
            \frac{\partial^2 D_{j}^{2}}{{\partial p_{j}}^2}
            &=(1-F(A))f^{\prime}(A+\Delta p)-F(p_i)f^{\prime}(p_j)+f(p_i)f(p_j)+\int_{p_{j}}^{A}f(u)f^{\prime}(u+\Delta p)du
    \end{aligned}
\]
\end{proof}


\section{The Property of the Implementable Prices Set} 
\label{appendix: set P}

\subsection{General Property}

Define a function $H: \mathbb{R}^2 \to \mathbb{R}$, 
    \[H(p_1, p_2) = p_2{\mathop{\max}\limits_{p_1^\prime} \pi_{1}^2(p_{1}^{\prime},p_{2})}  + p_1{\mathop{\max} \limits_{p_2^\prime}\pi_{2}^2(p_{2}^{\prime},p_{1})} + p_1p_2F(p_{1})F(p_{2}) - p_1p_2.\]

Then $P=\{(p_1,p_2):H(p_1,p_2)\le 0\}$. Since $H$ is continuous, $P$ is closed. Since $P$ is obviously bounded, $P$ is compact. The symmetry of $H$ over $p_1$ and $p_2$ contains that $P$ is symmetric about the line $p_1=p_2$.

By the proof of Proposition 2, we can see that, for all $(p_1,p_2) \in P$, $\varphi(p_1,p_2)$ is a compact set $[\underline{\alpha}(p_1,p_2), \bar{\alpha}(p_1,p_2)]$. $\varphi(p_1,p_2)$ reduces to a singleton if and only if $(p_1,p_2) \in \partial P$ where $\underline{\alpha}(p_1,p_2)=\bar{\alpha}(p_1,p_2)$.

As $s$ increases from $s_1$ to $s_2$, $A$ decreases and then $D_1^2$ and $D_2^2$ decreases since $$\frac{\partial D_1^2}{\partial A}=[1-F(A)]f(A-\Delta p)>0, \frac{\partial D_2^2}{\partial A}=[1-F(A)]f(A+\Delta p)>0$$
where $\Delta p=p_1-p_2$. Thus, $\pi_1^2(s_1)$ is higher than $\pi_1^2(s_2)$ given any $p_1^\prime$ and $p_2$, while $\pi_2^2(s_1)$ are higher than $\pi_2^2(s_2)$ given any $p_1$ and $p_2^\prime$. As a result, $\max \pi_1^2$ and $\max \pi_2^2$ is higher when $s=s_1$ than when $s=s_2$. Then, for any $(p_1,p_2)\in P_{s_1}$, $H(p_1,p_2|s_2)\le H(p_1,p_2|s_1)\le 0$, which means  $(p_1,p_2)\in P_{s_2}$. Thus, $P_{s_1}\subseteq P_{s_2}$

\subsection{Results under Uniform Distribution}

Consider an environment where \(F(v) \sim U[0,1]\). All consumers have the same search cost $s$. We thus have \(A = 1- \sqrt{2s}\).

Given $\boldsymbol{p} = (p_1,p_2)$, using the formulas calculated in Appendix A.1,  

\[D_1^1(\boldsymbol{p}) =\frac{1}{2}A^2-\frac{1}{2}p_2^2 -A-p_1+p_2 +1 \]
\[D_1^2(\boldsymbol{p}) = - \frac{1}{2} A^2 + \frac{1}{2}p_1^2 - p_1p_2 + A -p_1 + p_2\]
\[D_2^1(\boldsymbol{p}) = \frac{1}{2}A^2 - \frac{1}{2} p_1^2- A + p_1 - p_2  +1\]
\[D_2^2(\boldsymbol{p}) = -\frac{1}{2}A^2 + \frac{1}{2} p_2^2-p_1p_2 +A+p_1-p_2\]
\[B = D_1^1(\boldsymbol{p}) - D_1^2(\boldsymbol{p}) = D_2^1(\boldsymbol{p}) - D_2^2(\boldsymbol{p}) = (1-A)^2 - \frac{1}{2}(p_1-p_2)^2\]

We thus have the condition of $P$ as 
\[\frac{\max\limits_{p_1^\prime} p_1^\prime D_1^2(p_1^\prime,p_2)}{p_1} + \frac{\max\limits_{p_2^\prime} p_2^\prime D_2^2(p_1,p_2^\prime)}{p_2} + p_1p_2 \leq 1 \]

The first order condition of $\max \pi_2^2$ is
\[ \frac{3}{2}p_2^{\prime 2} -(2+2p_1)p_2^\prime + (p_1+A-\frac{1}{2} A^2)=0\]

Thus the best response of the deviated price $p_2^\prime\triangleq \hat{p}_2(p_1)$ satisfied the above function and we can induce that
\[ \frac{d\hat{p}_2}{dp_1}=\frac{1-2\hat{p}_2}{2+2p_1-3\hat{p}_2}\]

The implicit function $p_2^\prime= \hat{p}_2(p_1)$ have two roots, but only the the lowest root can be the best response by checking the second order condition. Thus, $\hat{p}_2<\frac{2+2p_1}{3}$, or $2+2p_1-3\hat{p}_2>0$.

A sufficient condition for $P$ to be convex is that the function $H$ is convex over $(p_1,p_2)\in P$. By the definition of $H$ and the uniform distribution,
\[ \frac{\partial H}{\partial p_1}= \max \pi_2^2+p_1[\max\pi_2^2]_{p_1}^\prime+2p_1p_2^2-p_2\]
\[ \frac{\partial^2 H}{\partial p_1^2}= 2[\max\pi_2^2]_{p_1}^\prime + p_1[\max\pi_2^2]_{p_1}^{\prime\prime} +2p_2^2\]
\[ \frac{\partial^2 H}{\partial p_1 \partial p_2}= 4p_1p_2-1\]
where by the envelope theorem,
\[ [\max\pi_2^2]_{p_1}^\prime= -\hat{p}_2^2+\hat{p}_2, [\max\pi_2^2]_{p_1}^{\prime\prime}=[1-2\hat{p}_2]\frac{d\hat{p}_2}{dp_1} \]

Thus, on one hand, 
\[ \frac{\partial^2 H}{\partial p_1^2}= 2\hat{p}_2(1-\hat{p}_2)+ p_1\frac{(1-2\hat{p}_2)^2}{2+2p_1-3\hat{p}_2}+2p_2^2>2p_2^2\ge 0 \]

With the same logic, we can prove that $\frac{\partial^2 H}{\partial p_2^2}>2p_1^2\ge 0 $

On the other hand, when $s\to 0$, or equivalently when $A\to 1$, the set $P$ shrinks and all $(p_1,p_2)\in P$ satisfied $1/6<p_1p_2<1/2$. Under this condition, 
\[\frac{\partial^2 H}{\partial p_1^2}\frac{\partial^2 H}{\partial p_2^2} - (\frac{\partial^2 H}{\partial p_1 \partial p_2})^2>2p_2^2\cdot 2p_1^2-(4p_1p_2-1)^2=(1-2p_1p_2)(6p_1p_2-1)>0\]

The social values under uniform distribution are
\[V_1^1=\frac{1}{2}[1-(A+\Delta p)^2]+\frac{1}{3}[(A+\Delta p)^3-p_1^3]-\frac{1}{2}\Delta p [(A+\Delta p)^2-p^2_1]\]
\[V_2^2=(A+\Delta p)\frac{1}{2}[1-A^2]+\frac{1}{3}[A^3-p_2^3]+\frac{1}{2}\Delta p [A^2-p^2_2]\]
\[V_1=\frac{1}{3}[A^3+(A+\Delta p)^3-p^3_1-p^3_2]+\frac{1}{2}[1-(A+\Delta p)^2]+\frac{\Delta p}{2}[1-(A+\Delta p)^2-p_2^2+p_1^2]+\frac{A}{2}[1-A^2]\]
\[V_1-V_2=-\frac{1}{3}\Delta p^3+(A-1)^2\Delta p\]

The social welfare contains the loss in search cost, which means,
\[SW_1=V_1-F(A+\Delta p)s,SW_2=V_2-F(A-\Delta p)s\]

Thus,
\[SW_1-SW_2=-\frac{1}{3}\Delta p^3+(A-1)^2\Delta p-2s\Delta p=-\frac{1}{3}\Delta p^3\]

\section{Proof of Main Results}
\label{appendix: proof of main results}

\paragraph{Proof of Lemma \ref{lemma: pin down alpha}}

Since $J^{\Pi}(p_1, p_2, \alpha) = \pi_{1}^{2}+\pi_{2}^{2}+ \alpha p_{1}B(p_1,p_2)+(1-\alpha) p_{2}B(p_1,p_2)$

\[\frac{\partial J^{\Pi}(p_1, p_2, \alpha)}{\partial \alpha} = (p_1-p_2)B(p_1,p_2)\]

Since $B>0$, when $p_1 > p_2$, $\frac{\partial J^{\Pi}(p_1, p_2, \alpha)}{\partial \alpha} > 0$; when $p_1 < p_2$, $\frac{\partial J^{\Pi}(p_1, p_2, \alpha)}{\partial \alpha} < 0$, and when $p_1 = p_2$, $\frac{\partial J^{\Pi}(p_1, p_2, \alpha)}{\partial \alpha} \equiv 0$. Hence the part for industry profit of Lemma \ref{lemma: pin down alpha} is proved. 

As for social welfare, we use $SW_1(p_1, p_2)$ and $SW_2(p_1, p_2)$ respectively to denote the social welfare when seller $1$ and $2$ are searched first. Under uniform distribution,
\[SW_1(p_1, p_2)  - SW_2(p_1, p_2) = -\frac{(p_1 - p_2)^3}{3}\]

Since $J^{SW} = \alpha SW_1(p_1, p_2) + (1 - \alpha) SW_2(p_1, p_2)$ 

\[\frac{\partial J^{SW}}{\partial \alpha} = SW_1(p_1, p_2)  - SW_2(p_1, p_2)= -\frac{(p_1 - p_2)^3}{3}\]

When $p_1 > p_2$, $\frac{\partial J^{SW}}{\partial \alpha} < 0$; when $p_1 < p_2$, $\frac{\partial J^{SW}}{\partial \alpha} > 0$, and when $p_1 < p_2$, $\frac{\partial J^{SW}}{\partial \alpha} \equiv 0$. Then the results of the lemma follows. $\blacksquare$

\paragraph{Proof of Proposition \ref{prop:first best}}

We prove this result under uniform distribution. 

    According to Lemma \ref{lemma: pin down alpha}, when $p_1 > p_2$,
    \[\sup_{\alpha \in [0,1]}J^{\Pi}(p_1, p_2, \alpha) = J^{\Pi}(p_1, p_2, 1) = \pi_1^1 + \pi_2^2\]

    The first order condition of the maximization problem is

    \begin{equation}
\frac{\partial(\pi_1^1+\pi_2^2)}{\partial p_1} = 0, \frac{\partial(\pi_1^1+\pi_2^2)}{\partial p_2} = 0
\label{FOCasymmetric}
\end{equation}

Under uniform distribution, the above can be calculated
\begin{equation}
    \frac{\partial(\pi_1^1+\pi_2^2)}{\partial p_1} = \frac{A^2}{2} - A + 1- 2p_1 + 2p_2 - \frac{3p_2^2}{2} =0
    \label{multi_foc1}
\end{equation}

\begin{equation}
   \frac{\partial(\pi_1^1+\pi_2^2)}{\partial p_2} = - \frac{A^2}{2}+ A + 2p_1 - 2p_2 - 3p_1 p_2  + \frac{3p_2^2}{2}=0 
\end{equation}

Taking the sum of these equations gives:

\[1 - 3p_1p_2 = 0\]

Plug $p_1 = \frac{1}{3p_2}$ into equation (\ref{multi_foc1}), we have

\[\gamma(p_2) = \frac{A^2}{2} - A + 1- \frac{2}{3p_2} + 2p_2 - \frac{3p_2^2}{2} =0\]

Since \( \gamma''(p_2) < 0 \) and \( \gamma'(\sqrt{3}/3) > 0 \), it follows that \( \gamma'(p_2) > 0 \) for \( p_2 \in (0, \sqrt{3}/3] \).

Additionally, since \( \gamma(\sqrt{3}/3) = (A - 1)^2/2 > 0 \) and \( \gamma(0.1) < 0 \), there exists a \( p_2^* \in (0.1, \sqrt{3}/3) \) such that $\gamma(p_2)=0$ holds. Moreover, $\gamma(1) = (A - 1)^2/2 + \frac{1}{3} > 0$, so $p_2^*$ is the only solution of $\gamma(p_2)=0$. Then by $3p_1p_2=1$, \( p^*_1 > \sqrt{3}/3 > p^*_2 \).

Similarly, when $p_1 < p_2$, the first order condition will yield $p^*_1 < \sqrt{3}/3 < p^*_2$. 

When \( p_1 = p_2 \), \( J^{\Pi} \) does not vary with \( \alpha \), so we can set \( \alpha = 1 \) ($\alpha = 0$), and the total profit remains \( \pi_1^1 + \pi_2^2 \) (or $\pi_2^1 + \pi_1^2$). However, since \( p_1 = p_2 \) does not satisfy the first-order conditions, the profit is lower than when \( p_1 > p_2 \) and when $p_1 < p_2$. Hence, it is optimal for a multi-product seller to adopt a contract with asymmetric prices. 

Finally, under constraint $p_1=p_2=p$, we can get the total profit $J^{\Pi}(p, p, 1)=\pi_1^1+\pi^2_2=p(1-p^2)$ and the first order condition tells that the optimal uniform price $p^*=\sqrt{3}/3$. $\blacksquare$

\subsection*{Proof of Theorem 1}
The line $p_1 = p_2$ divides $P$ into two regions. We define $\bar{P}:= \{(p_1, p_2) \in P: p_1 \leq p_2 \}$ and $\underline{P}:= \{(p_1, p_2) \in P: p_1 \geq p_2 \}$. Define two points $\bar{M}:= (\bar{p}, \bar{p})$ and $\underline{M} := (\underline{p}, \underline{p})$ and the line $\ell = \{ (p_1, p_2) \in P : p_1 = p_2 \}$ which segments $P$ into $\bar{P}$ and $\underline{P}$.

\begin{definition}
    
For the curve $\partial P$, we define the following extreme points: 
\begin{itemize}
    \item Leftmost point $M^L = (p_1^L, p_2^L)$, where $p_1^L = \min\{p_1: (p_1, p_2) \in \partial P\}$.
    \item Rightmost point $M^R = (p_1^R, p_2^R)$, where $p_1^R = \max\{p_1: (p_1, p_2) \in \partial P\}$.
    \item Lowest point $M^{\text{low}} = (p_1^{\text{low}}, p_2^{\text{low}})$, where $p_2^{\text{low}} = \min\{p_2: (p_1, p_2) \in \partial P\}$.
    \item Highest point $M^{\text{high}} = (p_1^{\text{high}}, p_2^{\text{high}})$, where $p_2^{\text{low}} = \max\{p_2: (p_1, p_2) \in \partial P\}$.
\end{itemize}
    
\end{definition}

When $P$ is strictly convex, the above four extreme points are uniquely defined. Considering the function $H$ we defined in Appendix B.1, Notice that $\partial P$ is characterized by the implicit function $H(p_1,p_2)=0$. The implicit function theorem states that

    \begin{itemize}
        \item Except for the four extreme points, for any point \( (p_1, p_2) \) where \( H(p_1, p_2) = 0 \), there exists an open set \( U \subset \mathbb{R} \) containing \( p_1 \) and a unique function \( g: U \to \mathbb{R} \) such that for all \( p_1 \in U \), \( p_2 = g(p_1) \) and \( H(p_1, g(p_1)) = 0 \). Similarly, there exists an open set \( V \subset \mathbb{R} \) containing \( p_2 \), along with a unique function \( h: V \to \mathbb{R} \), such that for any \( p_2 \in V \), \( p_1 = h(p_2) \) and \( H(h(p_2), p_2) = 0 \).
\begin{equation}
    \mathrm{d}p_2 = -\frac{\partial H/\partial p_{1}}
    {\partial H/\partial p_{2}} \mathrm{d}p_1
    \label{dp2}
\end{equation}

\item At $M^{\text{low}}$ and $M^{\text{high}}$, $\frac{\mathrm{d}p_2}{\mathrm{d}p_1} = 0$ and $\frac{\mathrm{d}p_1}{\mathrm{d}p_2}$ does not exist. At $M^L$ and $M^R$, $\frac{\mathrm{d}p_1}{\mathrm{d}p_2} = 0$ and $\frac{\mathrm{d}p_2}{\mathrm{d}p_1}$ does not exist.
    \end{itemize}

\begin{lemma}
    $\hat{\Pi}_2(p_1,p_2)\triangleq\Pi_2(p_1,p_2,\underline{\alpha}(p_1,p_2))$ is concave in $p_1$, i.e., $\frac{\partial^2 \hat{\Pi}_2(p_1,p_2)}{\partial p_1^2} < 0$ 
    \label{lemma: 2nd derivative}
\end{lemma}

\begin{proof}
Recall that $\underline{\alpha}=\frac{\max\pi_1^2 - \pi_1^2}{p_1 B}$ is obtained from IC1. Then,
\[
\begin{aligned}
       \hat{\Pi}_2(p_1,p_2) &= \pi_2^2+(1-\underline{\alpha}) p_2B\\
       &= \pi_2^2(p_1,p_2) + p_2 \frac{p_1B - \max\pi_1^2 + \pi_1^2}{p_1}\\
       & = p_2D_2^2+ p_2D_1^1 - \frac{p_2 \max\pi_1^2}{p_1}\\
       & =p_2(1-F(p_1)F(p_2))  - \frac{p_2 \max\pi_1^2}{p_1}\\
       & =p_2\left(1-p_1p_2  - \frac{ \max\pi_1^2}{p_1}\right)
\end{aligned}
\]

Thus, $\frac{\partial^2 \hat{\Pi}_2(p_1,p_2)}{\partial p_1^2}=- 2p_1^{-3}p_2 \max\pi_1^2 < 0$, where $\max\pi_1^2=\max_{p_1^\prime}\pi_1^2(p_1^\prime,p_2)$ is a function only with respect to $p_2$.
\end{proof}

\begin{lemma}
     $\frac{\partial \hat{\Pi}_2(p_1,p_2)}{\partial p_1} \geq 0$ for any $(p_1, p_2) \in \bar{P}$ if and only if $\bar{p} \leq p^*$.
    \label{lemma: 1st order partial deriv.}
\end{lemma}

\begin{proof}
According to Lemma \ref{lemma: 2nd derivative}, for every point $(p_1, p_2) \in \bar{P}$, there exists a point $(\tilde{p}_1, p_2) \in l \cup Arc(M^{\text{high}}, \bar{M})$ with $\tilde{p}_1>p_1$ such that $\frac{\partial \hat{\Pi}_2(p_1,p_2)}{\partial p_1} \geq \frac{\partial \hat{\Pi}_2(\tilde{p}_1,p_2))}{\partial p_1}$. In general, we have 
\[\frac{\partial \hat{\Pi}_2(p_1,p_2)}{\partial p_1}= p_2(-p_2+\frac{\max\pi_1^2}{p_1^2})=\frac{p_2}{p_1^2}(\max\pi_1^2-p^2_1p_2)\]

In $\ell$, since $p_1 = p_2 = p$, the value of $\hat{\Pi}_2$ can be reduced to 
\[\lambda(p)=\frac{\partial \hat{\Pi}_2(p_1,p_2)}{\partial p_1}|_{p_1=p_2=p}= -p^2+\frac{\max_{p_1^\prime}\pi_1^2(p_1^\prime,p)}{p}=-p^2+\max_{p_1^\prime}\frac{\pi_1^2(p_1^\prime,p)}{p}\]
where
\[\frac{\pi_1^2(p_1^\prime,p)}{p}=p_1^\prime(\frac{(A-\frac{1}{2}A^2)-(p^\prime_1-\frac{1}{2}p^{\prime2}_1)}{p}+1-p_1^\prime)\]
according to the formula of $D_1^2$ under uniform distribution.

Then by the envelop theorem,
\[\lambda^\prime(p)=-2p-\frac{\hat{p}_1(p)}{p^2}[(A-\frac{1}{2}A^2)-(\hat{p}_1(p)-\frac{1}{2}\hat{p}_1(p))]<0\]
where $\hat{p}_1(p)=\arg \max_{p^\prime_1}\pi_1^2(p_1^\prime,p)<A<1$. As a result, $\lambda(p)$ is decreasing with respect to $p$, which means $\frac{\partial \hat{\Pi}_2(p_1,p_2)}{\partial p_1}$ attains its minimum at $(\bar{p}, \bar{p})$ within $l$. Now in the whole $\bar{P}$, $\frac{\partial \hat{\Pi}_2(p_1,p_2)}{\partial p_1}$ can attain its minimum only at some points on $Arc(M^{\text{high}},\bar{M})$.

Next we consider points in $Arc(M^{\text{high}},\bar{M})$. According to implicit function theorem, any point in $Arc(M^{\text{high}},\bar{M})$ can be locally expressed by $(p_1, g(p_1))$. Notice that in $Arc(M^{\text{high}},\bar{M})$, $-1=g^\prime (\bar{p}) \le g^\prime (p_1)\le g^\prime(p_1^{\text{high}})=0$ due to the convexity of $P$. Consider $\eta (p_1)\triangleq \max_{p_1^\prime}\pi_1^2(p_1^\prime,g(p_1))-p^2_1g(p_1)$, $p_1^{\text{high}} \le p_1\le \bar{p}$. Then, by envelop theorem and noticing $g(p_1)=p_2\ge p_1$ within $\bar{P}$,
\[\eta^\prime (p_1)=(1-\hat{p}_1(g(p_1)))\hat{p}_1(g(p_1))g^\prime(p_1)-2p_1g(p_1)-p_1^2g^\prime(p_1)\le -p_1^2(1+g^\prime(p_1))\le 0\]
where the two inequality cannot be equality simultaneously. Thus, $\eta^\prime(p_1)<0$, which means $\eta (p_1)$ attains its minimum at $p_1=\bar{p}$. At $p_1=\bar{p}$,
\[\eta (\bar{p})=\max_{p_1^\prime}\pi_1^2(p_1^\prime,\bar{p})-\bar{p}^3=\frac{\bar{p}(1-\bar{p}^2)}{2} - \bar{p}^3  = \frac{\bar{p} - 3\bar{p}^3}{2}\]

Here we use the relationship $\max_{p_1^\prime}\pi_1^2(p_1^\prime,\bar{p})=\frac{\bar{p}(1-\bar{p}^2)}{2}$. Notice that at point $\bar{M}$, IC1 and IC2 are binding and $\varphi(\bar{p},\bar{p})$ is a singleton. Since IC1 and IC2 are symmetric when $p_1=p_2$, if $\alpha$ is feasible, $1-\alpha$ either, which means the singleton must be $\varphi(\bar{p},\bar{p})=\{\frac{1}{2}\}$. By binding IC1, 
\[\max_{p_1^\prime}\pi_1^2(p_1^\prime,\bar{p})=\frac{1}{2}\pi_1^1(\bar{p},\bar{p})+\frac{1}{2}\pi_1^2(\bar{p},\bar{p})=\frac{1}{2}J^\Pi(\bar{p},\bar{p},1)=\frac{\bar{p}(1-\bar{p}^2)}{2}\]

Then, we can see that $\eta(\bar{p})\ge 0$ if and only if $\bar{p} \leq p^*$. Thus, for all $(p_1,p_2)\in Arc(M^{\text{high}},\bar{M})$, 
\[\frac{\partial \hat{\Pi}_2(p_1,p_2)}{\partial p_1}=\frac{p_2}{p_1^2}\eta(p_1)\ge 0\]
if and only if $\bar{p} \leq p^*$. Then, for all $(p_1,p_2)\in \bar{P}$, $\frac{\partial \hat{\Pi}_2(p_1,p_2)}{\partial p_1}\ge 0$
if and only if $\bar{p} \leq p^*$.
\end{proof}

\paragraph{Proof of Theorem \ref{thm1} (\(\bar{p} \leq p^*\))} 

Consider the case where \(\bar{p} \leq p^*\) when \(p_1 \leq p_2\). According to Lemma \ref{lemma: pin down alpha}, we have \(\bar{\mu}^{\Pi}(p_1, p_2) = J^{\Pi}(p_1, p_2, \underline{\alpha})\) where
\[J^{\Pi}(p_1, p_2, \underline{\alpha})=\Pi_1(p_1,p_2,\underline{\alpha}(p_1,p_2))+\Pi_2(p_1,p_2,\underline{\alpha}(p_1,p_2))=\max\pi_1^2+\hat{\Pi}_2(p_1,p_2)\]
where we use the binding IC1 at $\underline{\alpha}$. Then, 
\[\frac{\partial \bar{\mu}^{\Pi}(p_1, p_2)}{\partial p_1}  = \frac{\partial \hat{\Pi}_2(p_1,p_2)}{\partial p_1}\]

Using Lemma \ref{lemma: 1st order partial deriv.}, for any point \((p_1, p_2) \in \bar{P}\), we know that $\frac{\partial \bar{\mu}^{\Pi}(p_1, p_2)}{\partial p_1}  \geq 0$. This means that the value of $\bar{\mu}^{\Pi}(p_1, p_2)$ increases monotonically in the horizontal direction. Consequently, for every \((p_1, p_2) \in \bar{P}\), there exists a point \((\tilde{p}_1, p_2) \in  Arc(M^{\text{high}},\bar{M}) \cup \ell\) such that \(\bar{\mu}^{\Pi}(\tilde{p}_1, p_2) \geq \bar{\mu}^{\Pi}(p_1, p_2)\).

In $\ell$, $p_1=p_2=p$, then
\[\bar{\mu}^{\Pi}(p, p)=[\max\pi_1^2+p_2(1-p_1p_2-\frac{\max\pi_1^2}{p_1})]_{p_1=p_2=p}=p-p^3\]
reaches its maximum at $p=\bar{p}$ as $0<\underline{p} \le p\le \bar{p}\le p^*$.

In $Arc(M^{\text{high}},\bar{M})$, by the second property of Proposition \ref{prop: P's property}, both IC1 and IC2 are binding on the boundary. Hence, the value of $\bar{\mu}^{\Pi}(p_1, p_2)$ can be reduced to 
\[
\zeta(p_1) = \max_{p_1'} \pi_1^2(p_1', g(p_1)) + \max_{p_2'} \pi_2^2(p_1, p_2'),p_1^{\text{high}}\le p_1\le\bar{p}_1.
\]

Thus by the envelop theorem, 
\[\zeta^\prime(p_1) = \hat{p}_1(g(p_1))(1-\hat{p}_1(g(p_1)))g^\prime(p_1)+\hat{p}_2(p_1)(1-\hat{p}_2(p_1)) = \frac{\partial \bar{\mu}^\Pi}{\partial p_2}g^\prime(p_1)+\frac{\partial \bar{\mu}^\Pi}{\partial p_1}
\]
where \(g^\prime(p_1)=-\frac{\partial H/\partial p_1}{\partial H/\partial p_2}\), $\hat{p}_1(p_2)=\arg\max_{p_1^\prime}\pi_1^2(p_1^\prime,p_2)$ and $\hat{p}_2(p_1)=\arg\max_{p_2^\prime}\pi_2^2(p_1,p_2^\prime)$.

We now aim to prove that $\zeta^\prime(p_1)\ge 0$ for every $(p_1,g(p_1))$ within $Arc(M^{\text{high}},\bar{M})$, i.e. 
\[\frac{\partial \bar{\mu}^\Pi/\partial p_1}{\partial \bar{\mu}^\Pi/\partial p_2} \ge \frac{\partial H/\partial p_1}{\partial H/\partial p_2}\]
\[\Leftrightarrow \frac{\hat{p}_2(p_1)(1-\hat{p}_2(p_1))}{\hat{p}_1(g(p_1))(1-\hat{p}_1(g(p_1)))} \ge \frac{\max\pi_2^2+p_1\hat{p}_2(p_1)(1-\hat{p}_2(p_1))+2 p_1 g(p_1)^2-g(p_1)}{\max\pi_1^2+g(p_1)\hat{p}_1(g(p_1))(1-\hat{p}_1(g(p_1)))+2 p_1^2 g(p_1)-p_1}\]
or to express it in a reduced form, for every $(p_1,p_2)\in Arc(M^{\text{high}},\bar{M})$, let $p_1^*=\hat{p}_1(p_2)$ and $p_2^*=\hat{p}_2(p_1)$, then we aim to prove,

\begin{equation}
  \frac{p_2^*(1-p_2^*)}{p_1^*(1-p_1^*)} \geq \frac{\max \pi_2^2 + p_1p_2^*(1-p_2^*)-p_2(1-2p_1p_2)}{\max \pi_1^2 + p_2p_1^*(1-p_1^*)-p_1(1-2p_1p_2)} 
  \label{inequ of zeta}
\end{equation}

We use the following three steps to prove the equation (\ref{inequ of zeta}).

\textbf{Step 1}. Consider the specific form of the maximization problem $\max_{p_2} \pi_2^2(p_1, p_2) = p_2(A - \frac{A^2}{2} + p_1 - p_2 - p_1 p_2 + \frac{p_2^2}{2})$, the first order condition tells,
\begin{equation}
  \frac{\partial \pi_2^2}{\partial p_2}|_{p_2=p_2^*} = A - \frac{A^2}{2} + p_1 - 2p_2^* - 2p_1 p_2^*  + \frac{3p_2^{*2}}{2} = 0 
  \label{eq:first_deriv_pi_2_2}
\end{equation}
and plug it into the maximization problem, yielding
\[\max\pi_2^2 = p_2^*(p_1p_2^* + p_2^* - p_2^{*2})\]

Differential equation (\ref{eq:first_deriv_pi_2_2}) with respect to $p_1$ to get
\[\frac{dp_2^*}{dp_1}=\frac{1-2p_2^*}{2+2p_1-3p_2^*}\]

Since $\frac{\partial \pi_2^2}{\partial p_2}|_{p_2=1/2}=A-\frac{A^2}{2}-\frac{5}{8}<0$, we have $p_2^*<\frac{1}{2}$. Thus, $\frac{d p_2^*}{d p_1}>0$. Then, by $p_1<p_2$ and the symmetric formulas of $p_1^*(p_2)$ and $p_2^*(p_1)$, we know that $p_1^*>p^*_2$.

As a result,
\[
\begin{aligned}
& \frac{p_2^*(1-p_2^*)}{p_1^*(1-p_1^*)} \geq \frac{\max \pi_2^2}{\max \pi_1^2} = \frac{p_2^*(p_1p_2^* + p_2^* - p_2^{*2})}
{p_1^*(p_2p_1^* + p_1^* - p_1^{*2})}\\
\Leftrightarrow & (1-p_2^*)(p_2p_1^* + p_1^* - p_1^{*2})  \geq (1-p_1^*)(p_1p_2^* + p_2^* - p_2^{*2})\\
\Leftrightarrow & (p_1^*-p_2^*)(1-p_1^*)(1-p_2^*)+(p_2-p_1)(1-p_1^*p_2^*)+p_2(p^*_1-p_2^*)\ge 0\\
\end{aligned}
\]
which is true by the relationships $p_1<p_2$ and $p_1^*>p_2^*$.


\textbf{Step 2}. Since $p_1<p_2$, we have
\[\frac{p_2^*(1-p_2^*)}{p_1^*(1-p_1^*)} \geq \frac{p_1p_2^*(1-p_2^*)}{p_2p_1^*(1-p_1^*)}\]

Now use the mediant inequality, we have
\[\frac{p_2^*(1-p_2^*)}{p_1^*(1-p_1^*)} \geq \frac{\max \pi_2^2+p_1p_2^*(1-p_2^*)}{\max \pi_1^2+p_2p_1^*(1-p_1^*)}\]

\textbf{Step 3}. Because in $Arc(M^{\text{high}},\bar{M})$, both $\partial H /\partial p_1$ and $\partial H / \partial p_2$ are positive, we have
\[
\begin{aligned}
    \frac{\max \pi_2^2+p_1p_2^*(1-p_2^*)}{\max \pi_1^2+p_2p_1^*(1-p_1^*)} &\geq \frac{\max \pi_2^2 + p_1p_2^*(1-p_2^*)-p_1(1-2p_1p_2)}
    {\max \pi_1^2 + p_2p_1^*(1-p_1^*)-p_1(1-2p_1p_2)} \\
    &\geq \frac{\max \pi_2^2 + p_1p_2^*(1-p_2^*)-p_2(1-2p_1p_2)}
    {\max \pi_1^2 + p_2p_1^*(1-p_1^*)-p_1(1-2p_1p_2)}
    \end{aligned}
\]

We thus have proved that $\zeta^\prime(p_1)\ge 0$. When $p_1<p_2$, we have $\zeta^\prime(p_1)>0$. Thus, $\bar{\mu}(p_1,p_2)$ can reach its maximum only at point $(\bar{p},\bar{p})$ within $Arc(M^{\text{high}},\bar{M})$, and then within $\bar{P}$.

A similar proof holds for \(\underline{P}\) due to the symmetry of \(P\). Therefore, the proof of the statement when $\bar{p}\leq p^*$ in Theorem \ref{thm1} is finished. $\blacksquare$

\paragraph{Proof of Theorem \ref{thm1}($\bar{p} > p^*$)}
To show the optimal contract is asymmetric, we need to show two things: first, the optimal contract cannot take a form of $\alpha = 1/2$; second, the optimal contract does not take the form $p_1 = p_2$, where $\alpha$ is irrelevant and can be set to be $1/2$. 

\textbf{Step 1. interior $(p_1,p_2,1/2),p_1\ne p_2$ cannot be optimal.} Consider the case $p_1<p_2$, suppose $\underline{\alpha}=\frac{1}{2}$ and $(p_1,p_2,\underline{\alpha})$ is optimal, by the first order condition, a necessary condition for an interior price vector $(p_1,p_2)$ to be optimal is \(\frac{\partial \bar{\mu}^\Pi}{\partial p_1} =0\) and \(\frac{\partial \bar{\mu}^\Pi}{\partial p_2} =0\). Since $\bar{\mu}^\Pi=\max\pi_1^2+\hat{\Pi}_2$, then
\[\frac{\partial \bar{\mu}^\Pi}{\partial p_1}=\frac{\partial \hat{\Pi}_2}{\partial p_1}=\frac{p_2}{p_1^2}(\max\pi_1^2-p^2_1p_2)=0\]
\[\frac{\partial \bar{\mu}^\Pi}{\partial p_2} =[\max\pi^2_1]^\prime_{p_2}+1-2p_1p_2-\frac{\max\pi^2_1}{p_1}=0\]

Thus,
\[x\triangleq 3p_1p_2-1=[\max\pi^2_1]^\prime_{p_2}=(1-p_1^*)p_1^*>0\]

Meanwhile, $\underline{\alpha}=\frac{1}{2}$ means IC1 binding for $\alpha=\frac{1}{2}$, i.e. 
\[\max\pi_1^2=\frac{1}{2}(\pi_1^1+\pi_1^2)=\frac{p_1}{2}(\frac{p_1^2}{2}-\frac{p_2^2}{2}-p_1p_2+1-2p_1+2p_2)\]

Substitute it into $\max\pi_1^2-p^2_1p_2=0$ to get
\[\delta (p_1,p_2)=\frac{p_1^2}{4}-\frac{p_2^2}{4}-\frac{3}{2}p_1p_2+\frac{1}{2}-p_1+p_2=0\]

Since $(p_1,p_2)$ satisfied $\delta (p_1,p_2)=0$, we can differential it to get 
\[\frac{dp_2}{dp_1}=\frac{p_1-3p_2-2}{p_2+3p_1-2}=\frac{2p_2(p_1-3p_2-2)}{p_1^2+p_2^2+2(1-2p_1)}<0, \forall p_1\in [1/2,p^*]\]

One can verify that $(p^*,p^*)$ satisfied $\delta (p_1,p_2)=0$. Now consider the part of $\delta(p_1,p_2)$ possibly inside $\bar{P}$. By Figure \ref{fig:boundary}, when $A=0.7$, $p_2^{\text{high}}<0.7$. Thus, for all possible $A$, $p_2\le \min\{A,p_2^{\text{high}}(A)\} <0.7$ as $P$ is shrinking with respect to the increase of $A$. As a result, we only need to pay attention to $p_2\in [p^*,0.7]$ within $\delta=0$ and then $p_1\in [0.56,p^*]$ where $\frac{dp_2}{dp_1}>0$. Within $\delta=0$, we have
\[x(p_1)=3p_1p_2-1=\frac{1}{2}(p_1^2-p_2^2)+2(p_2-p_1)\]
\[x^\prime(p_1)=(p_1-2)+(2-p_2)\frac{dp_2}{dp_1}<0\]
which means $x\in [0,x(0.56)]\subset [0,0.2]$. Then by $x=(1-p^*_1)p^*_1$, we can induce that $p_1^*\in [0,0.3]$. Thus, $\max\pi_1^2=p_1^{*2}(p_2+1-p_1^*)\le 0.13$ since $\frac{\partial \max\pi_1^2}{\partial p_1^*}=2p_1^*(p_2+1)-3p_1^{*2}>p_1^*(2-3p_1^*)>0$.

On the other hand, $y(p_2)=\max\pi_1^2=p_1^2p_2$ defined within $\delta (p_1,p_2)=0$ satisfied
\[
\begin{aligned}
    y^\prime(p_2)& =\frac{1}{3}p_1x^\prime(p_1)+p_1p_2=\frac{p_1}{3}[(3p_2+p_1-2)+(2-p_2)\frac{dp_2}{dp_1}] \\
    & =\frac{p_1}{3}[(p_2-3p_1-2)\frac{1}{dp_1/dp_2}+(2-p_2)\frac{dp_2}{dp_1}]=-p_1^2\frac{dp_2}{dp_1}>0
    \end{aligned}
\]

Thus, $\max\pi_1^2\ge p^{*3}\approx 0.19$, a contradiction.

\textbf{Step 2. boundary $(p_1,p_2,1/2),p_1\ne p_2$ cannot be optimal.} We only need to show that for all $(p_1,p_2)\in \partial \bar{P}-\ell$, only $\bar{M}$ can reach the maximum of $\bar{\mu}^\Pi(p_1,p_2)$. We have proved in the case $\bar{p}<p^*$ that  $\bar{\mu}(p_1,p_2)$ can reach its maximum only at point $(\bar{p},\bar{p})$ within $Arc(M^{\text{high}},\bar{M})$. Since this proof is not based on $\bar{p}<p^*$ and $\frac{\partial \bar{\mu}^\Pi/\partial p_1}{\partial \bar{\mu}^\Pi/\partial p_2} \ge \frac{\partial H/\partial p_1}{\partial H/\partial p_2}$ holds obviously since $\partial H/\partial p_1<0$ and $\partial H/\partial p_1>0$ within $Arc(M^L,M^{\text{high}})$ and have the same monotonicity within $Arc(\underline{M},M^L)$ by a similar proof within $Arc(M^{\text{high}},\bar{M})$, we can learn that $\bar{\mu}(p_1,p_2)$ can reach its maximum only at point $(\bar{p},\bar{p})$ within $\partial \bar{P}-\ell$.


\textbf{Step 3. $(p_1,p_2,1/2),p_1=p_2$ cannot be optimal.} It is easy to know that $\bar{\mu}^\Pi(p,p)=p-p^3$ reaches its maximum at $p=p^*$ in $\ell$ as $0<p^*<\bar{p}$. If $(p^*,p^*,\frac{1}{2})$ is also the optimal contract, then since its an interior point, we have
\[\frac{\partial \bar{\mu}^\Pi}{\partial p_1}=\frac{p_2}{p_1^2}(\max\pi_1^2-p^2_1p_2)=\frac{\pi_1^2(p_1^*,p^*)}{p^*}-p^{*2}=0\]
\[\frac{\partial \bar{\mu}^\Pi}{\partial p_2} =[\max\pi^2_1]^\prime_{p_2}+1-2p_1p_2-\frac{\max\pi^2_1}{p_1}=(p^*_1-p_1^{*2})+1-2p^{*2}-\frac{\pi_1^2(p_1^*,p^*)}{p^*}=0\]

Take the sum to get
\[p_1^*-p_1^{*2}=3p^{*2}-1=0\Rightarrow p_1^*=0 ~\text{or}~ p_1^*=1, \]
which induce a contradiction since $\pi_1^2(0,p^*)=\pi_1^2(1,p^*)=0<\pi_1^2(p^*,p^*)$.

\subsection*{Proof of Proposition \ref{prop:trade probability} and \ref{prop: welfare max}}

\paragraph{Proof of Proposition \ref{prop:trade probability}}

Under the uniform distribution, the trade probability (i.e., total demand) is given by
\[
J^{TP}(p_1, p_2) = 1 - p_1 p_2.
\]

We first consider $\underline{P}$. Since \(\frac{\partial J^{TP}(p_1, p_2)}{p_1} = -p_2 < 0\), for every point $(p_1, p_2) \in \underline{P}$, there exists a point $(\tilde{p}_1, p_2) \in \ell \cup Arc(\underline{M}, M^{\text{low}})$ such that $J^{TP}(\tilde{p}_1, p_2) \geq J^{TP}(p_1, p_2)$ 

In $\ell$, $J^{TP} = 1-p^2$, thus $J^{TP}$ attains its maximum at $(\underline{p},\underline{p})$.

Similar to the process in the proof of Theorem \ref{thm1}, we have 
\[\frac{\mathrm{d}J^{TP}(p_1, g(p_1))}{\mathrm{d} p_1} = \frac{\partial J^{TP}}{\partial p_1} + \frac{\partial J^{TP}}{\partial p_2}g^\prime(p_1) = -p_2 -p_1\left(-\frac{\partial H / \partial p_1}{\partial H / \partial p_2}\right).\]

It suffices to show for every point $(p_1, p_2) \in Arc(\underline{M}, M^{\text{low}})$

\[\frac{\mathrm{d}J^{TP}(p_1, g(p_1))}{\mathrm{d} p_1} = -p_2 -p_1\left(-\frac{\partial H / \partial p_1}{\partial H / \partial p_2}\right) \leq 0\]

It is implied by

\[\frac{p_2}{p_1} \geq \frac{\partial H/\partial p_1}{\partial H/\partial p_2} = \frac{\max \pi_2^2 + p_1p_2^*(1-p_2^*)-p_2(1-2p_1p_2)}
    {\max \pi_1^2 + p_2p_1^*(1-p_1^*)-p_1(1-2p_1p_2)}\]

Note that in $Arc(\underline{M}, M^{\text{low}})$, $\partial H / \partial p_1$ and $\partial H / \partial p_2$ are both negative. It suffices to show that
$$
\begin{aligned}
    & p_2 (\max \pi_1^2 + p_2p_1^*(1-p_1^*)-p_1(1-2p_1p_2)) \leq p_1(\max \pi_2^2 + p_1p_2^*(1-p_2^*)-p_2(1-2p_1p_2))\\
    \Leftrightarrow ~ & [p_2p_1^{*2}(p_2+1-p_1^*)-p_1p_2^{*2}(p_1+1-p_2^*)] + [p_2^2p_1^*(1-p_1^*)- p_1^2p_2^*(1-p_2^*)]\le 0
\end{aligned}
$$

Notice that we have proved $\frac{dp^*}{dp}<0$ and $p^*<\frac{1}{2}$, which induces that $p^*(1-p^*)$ is increasing with respect to $p^*$. Then, by $p_1>p_2$, we learn that $p_1^*<p_2^*$ and $p_1^*(1-p_1^*)<p_2^*(1-p_2^*)$. Thus,
\[(p_2^2p_1^{*2}-p_1^2p_2^{*2})+[p_2p_1^*\cdot p_1^*(1-p_1^*)-p_1p_2^*\cdot p_2^*(1-p_2^*)]+[p_2^2\cdot p_1^*(1-p_1^*)-p_1^2\cdot p_2^*(1-p_2^*)]\le 0,\]
which means $\frac{p_2}{p_1} \geq \frac{\partial H/\partial p_1}{\partial H/\partial p_2}$ holds. And if  $p_1 > p_2$, the inequality above will become strict. Thus, we proved trade probability attains its maximum only at $(\underline{p}, \underline{p})$ in $\underline{P}$. By symmetry of $P$, we finish the proof.

\paragraph{Proof of Proposition \ref{prop: welfare max}}

\paragraph{Industry profit.} We first prove that the industry profit is minimized at point $(\underline{p}, \underline{p})$. According to Lemma \ref{lemma: pin down alpha}, when $p_1 \geq p_2$, $\underline{\mu}^{\Pi}(p_1, p_2) = J^{\Pi}(p_1, p_2, \underline{\alpha})$. Same as proving Theorem \ref{thm1}, 

\[
\frac{\partial \underline{\mu}^{\Pi}(p_1, p_2)}{\partial p_1}  = \frac{\partial (\Pi_1 + \Pi_2)}{\partial p_1} = \frac{\partial \hat{\Pi}_2(p_1, p_2)}{\partial p_1}.
\]

According to Lemma \ref{lemma: 2nd derivative}, we have for every point $(p_1, p_2) \in \underline{P}$, there exists a point $(\tilde{p}_1, p_2)$ in $\partial \underline{P}$ such that $\hat{\Pi}_2(p_1, p_2) \geq \hat{\Pi}_2(\tilde{p}_1, p_2)$, which implies $\underline{\mu}^{\Pi}(p_1, p_2) \geq \underline{\mu}^{\Pi}(\tilde{p}_1, p_2)$. Next we prove that $\underline{\mu}^{\Pi}(p_1, p_2)$ reaches its minimum on $\partial \underline{P} \cup \ell$ at the point $(\underline{p}, \underline{p})$.

In $\ell$, the only candidate that might be the minimizer industrial profit than $(\underline{p}, \underline{p})$ is $(\bar{p}, \bar{p})$. To see this, notice that $\underline{\mu}^{\Pi} = p(1-p^2)$ when $p_1 = p_2 = p$. When $p^* \geq \bar{p}$, $\frac{\partial \underline{\mu}^{\Pi}}{\partial p} = 1 - 3p^2 \geq 0$, $\underline{\mu}^{\Pi}$ attains its minimum at $(\underline{p}, \underline{p})$. When $p^* > \bar{p}$, we need to check industry profit at point $(\bar{p}, \bar{p})$. For this, we make the following claim: 

\begin{claim}
   $\underline{\mu}^{\Pi}(\bar{p},\bar{p}) \geq \underline{\mu}^{\Pi}(\underline{p},\underline{p})$, the equation holds if and only if $s=0$.
   \end{claim}
\begin{proof}
     For $p \in \partial P \cup \ell$, since the point is on the boundary, we have $\frac{2\max_{p^\prime}\pi(p^\prime,p)}{p}+F(p)^2=1$ and $\underline{\mu}^{\Pi}(p,p)=p(1-F(p)^2)$, then $\underline{\mu}^{\Pi}(p,p)=2\max_{p^\prime}\pi(p^\prime,p)$. We have $\Pi(\bar{p},\bar{p}) \geq \Pi(\underline{p},\underline{p})$ since $\max_{p_1^\prime}\pi(p_1^\prime,p)$ increases with $p$ (the derivative is $p_1^*(1-p_1^*)>0$). 
\end{proof}

Therefore, in $\ell$, $\underline{\mu}^{\Pi}$ attains its minimum at $(\underline{p}, \underline{p})$.

Next we consider $\partial P$. In $\partial P$, we need to consider three cases: $Arc(M^r, \bar{M})$, $Arc(M^{\text{low}}, M^r)$ and $Arc( \underline{M}, M^{\text{low}})$. Our goal is to find out the minimizer of $\underline{\mu}^{\Pi}(p_1, p_2)$ in each arc. According to the implicit function theorem, points in each arc can be expressed as $(p_1(p_2), p_2)$ with $\frac{dp_1}{dp_2}=-\frac{\partial H/\partial p_2}{\partial H / \partial p_1}$ or $(p_1, p_2(p_1))$ with $\frac{dp_2}{dp_1}=-\frac{\partial H/\partial p_1}{\partial H / \partial p_2}$. Then, 
\[\frac{\mathrm{d}\underline{\mu}^{\Pi}}{\mathrm{d}p_2} = \frac{\partial \underline{\mu}^{\Pi}}{\partial p_1} \left(-\frac{\partial H / \partial p_2}{\partial H / \partial p_1}\right) + \frac{\partial \underline{\mu}^{\Pi}}{\partial p_2} \ge 0 \Leftrightarrow \frac{\partial \underline{\mu}^{\Pi}/\partial p_2}{\partial \underline{\mu}^{\Pi}/\partial p_1} \ge \frac{\partial H/\partial p_{2}}{\partial H/\partial p_{1}}\]
\[\frac{\mathrm{d}\underline{\mu}^{\Pi}}{\mathrm{d}p_1} = \frac{\partial \underline{\mu}^{\Pi}}{\partial p_1}  + \frac{\partial \underline{\mu}^{\Pi}}{\partial p_2}\left(-\frac{\partial H / \partial p_1}{\partial H / \partial p_2}\right) \ge 0 \Leftrightarrow \frac{\partial \underline{\mu}^{\Pi}/\partial p_1}{\partial \underline{\mu}^{\Pi}/\partial p_2} \ge \frac{\partial H/\partial p_{1}}{\partial H/\partial p_{2}}\]

In $Arc(M^r, \bar{M})$, $p_1\ge p_2$, , $\partial H / \partial p_1 > 0$ and $\partial H / \partial p_2 > 0$, thus $\frac{\mathrm{d}\underline{\mu}^{\Pi}}{\mathrm{d}p_2}\ge 0$ is equivalent to
\[
\frac{p_1^*(1-p_1^*)}{p_2^*(1-p_2^*)} \geq \frac{\max \pi_1^2 + p_2p_1^*(1-p_1^*)-p_1(1-2p_2p_1)}
    {\max \pi_2^2 + p_1p_2^*(1-p_2^*)-p_2(1-2p_2p_1)}
\]
which holds due to the similar method as in proving Theorem $\ref{thm1}$. Thus, in $Arc(M^r, \bar{M})$, $\underline{\mu}^{\Pi}$ attains its minimum at point $M^r$.

In $Arc(M^{\text{low}}, M^r)$, since $\frac{dp_1}{dp_2}=-\frac{\partial H/\partial p_2}{\partial H / \partial p_1}>0$, $\partial \underline{\mu}^{\Pi}/\partial p_1=p_2^*(1-p_2^*)>0$ and $\partial \underline{\mu}^{\Pi}/\partial p_2=p_1^*(1-p_1^*)>0$, we learn that $\frac{\mathrm{d}\underline{\mu}^{\Pi}}{\mathrm{d}p_2}>0$, which means $\underline{\mu}^{\Pi}$ attains its minimum at point $M^{\text{low}}$.

In $Arc(M^r, \bar{M})$, $p_1\ge p_2$, , $\partial H / \partial p_1 < 0$ and $\partial H / \partial p_2 < 0$, thus $\frac{\mathrm{d}\underline{\mu}^{\Pi}}{\mathrm{d}p_1}\ge 0$ is equivalent to
\[
\frac{p_2^*(1-p_2^*)}{p_1^*(1-p_1^*)} \geq \frac{-\max\pi_2^2 - p_1p_2^*(1-p_2^*)+p_2(1-2p_2p_1)}
    { -\max \pi_1^2 - p_2p_1^*(1-p_1^*)+ p_1(1-2p_2p_1)}
\]
which holds by noticing that
\[p_2^*(1-p_2^*)\ge p_1^*(1-p_1^*)\]
\[(p_1^{*2}p_2-p_2^{*2}p_1)+[p_1^*\cdot p_1^*(1-p_1^*)-p_2^*\cdot p_2^*(1-p_2^*)] + [p_2\cdot p_1^*(1-p_1^*)-p_1\cdot p_2^*(1-p_2^*)]+ (p_2-p_1)(1-2p_2p_1)\le 0\]

Thus, in $Arc(M^r, \bar{M})$, $\underline{\mu}^{\Pi}$ attains its minimum at point $\underline{M}$.

Finally, it is easy to show the equality holds only at point $(\underline{p}, \underline{p})$. Thus, in $\partial \underline{P}$, $J^{\Pi}$ attains its minimum at $(\underline{p}, \underline{p})$. By the symmetry of $P$, this result can also be proved for $\partial \bar{P}$, which completes the proof. $\blacksquare$

\paragraph{Social Welfare}

\begin{claim}
    For every $(p_1, p_2) \in \underline{P}$, we have $J^{SW}(p_1, p_2, \underline{\alpha}) \leq J^{SW}(p_2, p_2, \underline{\alpha})$
\end{claim}

\begin{proof}
We first give the reduced form of $\delta(p_1,p_2)\triangleq J^{SW}(p_2, p_2, \underline{\alpha})-J^{SW}(p_1, p_2, \underline{\alpha})$. By Appendix B.2, under uniform distribution, we learn that,
\[
\begin{aligned}
    J^{SW}(p_1, p_2, \underline{\alpha})= ~ &  \underline{\alpha}SW_1+(1-\underline{\alpha})SW_2=SW_2+\underline{\alpha}(SW_1-SW_2)\\
    = ~& \frac{1}{3}[A^3+(A-\Delta p)^3-p^3_1-p^3_2]+\frac{1}{2}[1-(A-\Delta p)^2]-\frac{\Delta p}{2}[1-(A-\Delta p)^2+p_2^2-p_1^2]\\
    & +\frac{A}{2}[1-A^2]-(A-\Delta p)\frac{1}{2}(1-A)^2-\frac{\underline{\alpha}}{3}\Delta p^3\\
\end{aligned}
\]
where $\Delta p =p_1-p_2>0$.

Thus, set $p_1=p_2$ to get
\[
J^{SW}(p_2, p_2, \underline{\alpha})=\frac{1}{3}[A^3+A^3-p^3_2-p^3_2]+\frac{1}{2}[1-A^2]+\frac{A}{2}[1-A^2]-A\cdot\frac{1}{2}(1-A)^2
\]   

Then, we have
\[
\begin{aligned}
    \delta (p_1,p_2)= ~ &  \frac{1}{3}[A^3-(A-\Delta p)^3-p^3_2+p_1^3]+\frac{1}{2}[(A-\Delta p)^2-A^2] -\Delta p\cdot \frac{1}{2}(1-A)^2 \\
    & + \frac{\Delta p}{2}[1-(A-\Delta p)^2+p_2^2-p_1^2]+ \frac{\underline{\alpha}}{3}\Delta p^3\\
    = ~ & \frac{\Delta p}{6}[2\underline{\alpha}\Delta p^2+4p_2^2-2p_1^2+4p_1p_2+3\Delta p]
\end{aligned}
\]

Our aim is to prove that $\delta(p_1,p_2)\ge0$. Since $p_1>p_2$, define $e=p_1/p_2>1$, then
\[
4p_2^2-2p_1^2+4p_1p_2+3\Delta p=(-2e^2+4+4e)p^2_2+3(e-1)p_2 \ge (-2e^2+7e+1)p^2_2 \geq 0.
\]
holds for $0<e<3.5$. For sufficient low $s$, we have $p_1^R<0.7$ and $p_2^{\text{low}}>0.2$. Thus, $e=p_1/p_2<p_1^R/p_2^{\text{low}}=3.5$, which completes the proof. 
\end{proof}

Now we need to address the points where \( p_2 < \underline{p} \) and \( p_1 > p_2 \). This region is defined as
\[
G := \{(p_1, p_2) \in P : p_2 < \underline{p} \text{ and } p_1 > p_2 \}.
\]
The next result shows that for any point in this region, we can move every point leftward to the boundary.

\begin{claim}
Suppose that Assumption \ref{assumption: uniform} and \ref{assumption: A large} hold. For any point \( (p_1, p_2) \in G \), we have \( \frac{\partial \bar{\mu}^{SW}}{\partial p_1} \leq 0 \).
\label{negative partial derivative}
\end{claim}

\begin{proof}
We compute the partial derivative of \( \underline{\mu}^{SW} \) with respect to \( p_1 \) as
\[
\begin{aligned}
\frac{\partial \bar{\mu}^{SW}}{\partial p_1}=~& -[(A-\Delta p)^2+p_1^2]+(A-\Delta p)-\frac{1}{2}[1-(A-\Delta p)^2+p_2^2-p_1^2] \\
& -\Delta p(A-\Delta p-p_1)+\frac{1}{2}(1-A)^2-\underline{\alpha}\Delta p^2-\frac{1}{3}\Delta p^3\frac{\partial \underline{\alpha}}{\partial p_1}\\
=~& p_1^2-2p_1p_2-p_1+p_2-\underline{\alpha}\Delta p^2-\frac{1}{3}\Delta p^3\frac{\partial \underline{\alpha}}{\partial p_1} \\
\le ~& p_1^2-2p_1p_2-p_1+p_2-\frac{1}{3}\Delta p^3\frac{\partial \underline{\alpha}}{\partial p_1}
\end{aligned}
\]

Given \( (p_1, p_2) \), the function \( \underline{\alpha} \) is determined by the implicit condition IC1 as
\[
\text{IC1: } \max_{p_1'} p_1' D_1^2(p_1', p_2) = p_1 D_1^2(p_1, p_2) + \underline{\alpha}(p_1, p_2)p_1 B(p_1, p_2)).
\]

Taking partial derivatives with respect to \( p_1 \) from both sides, we get
\[
0 = D_1^2 + p_1 \frac{\partial D_1^2}{\partial p_1} + \frac{\partial \underline{\alpha}}{\partial p_1} p_1 B+ \underline{\alpha}B + \underline{\alpha} p_1 \frac{\partial B}{\partial p_1}.
\]

Thus,
\[
\begin{aligned}
\frac{\partial \underline{\alpha}}{\partial p_1}&= \frac{-1}{p_1B}[D_1^2 + p_1 \frac{\partial D_1^2}{\partial p_1} + \underline{\alpha}B + \underline{\alpha} p_1 \frac{\partial B}{\partial p_1}] \\
&\ge \frac{D_1^2 + p_1 \frac{\partial D_1^2}{\partial p_1} + (D_1^1 - D_1^2)}{-p_1B} \quad (\underline{\alpha} \leq 1, D_1^2 \leq D_1^1) \\
&= \frac{\frac{A^2}{2} - A + p_1^2 - p_1 p_2 - 2 p_1 - \frac{p_2^2}{2} + p_2 + 1}{-p_1B} \\
&\geq \frac{\frac{A^2}{2} - A -2p_1+ p_2 + 1 + p_1^2}{-p_1B} \\
&\geq \frac{\frac{3}{2} - A -2p_1+ p_2 + p_1^2}{-p_1B} \quad (A \leq 1, p_1 \leq 1).\\
& = \frac{\frac{3}{2} - A -2p_1+ p_2 + p_1^2}{p_1[\frac{1}{2}(p_1-p_2)^2 - (1-A^2)]}\\
& \geq \frac{\frac{3}{2} - A -2p_1+ p_2 + p_1^2}{-\frac{1}{2}p_1(1-A^2)} \quad (|p_1 - p_2| \leq 1-A)\\
& =  \frac{3 - 2A -4p_1+ 2p_2 + 2p_1^2}{-p_1(1-A^2)} 
\end{aligned}
\]


Therefore,
\[
\begin{aligned}
\frac{\partial \bar{\mu}^{SW}}{\partial p_1} 
&\leq \frac{3 - 2A -4p_1+ 2p_2 + 2p_1^2}{-p_1(1-A^2)} \left( \frac{-(\Delta p)^3}{3} \right) + p_2 - p_1 - 2 p_1 p_2 + p_1^2 \\
&= \frac{(3 - 2A -4p_1+ 2p_2 + 2p_1^2)\Delta p}{3p_1}  + p_2 - p_1 - 2 p_1 p_2 + p_1^2  \quad (|p_1 - p_2| \leq 1-A)\\
& = \frac{(3-2A)\Delta p + 5p_1^3 - 8p_1^2p_2 - 7p_1^2 + 9p_1p_2 - 2p_2^2}{3p_1}
\end{aligned}
\]

Let $\xi = (3-2A)\Delta p + 5p_1^3 - 8p_1^2p_2 - 7p_1^2 + 9p_1p_2 - 2p_2^2$. To prove the claim it suffices to show $\xi \leq 0$ for every point in $G$ when $A$ is large enough. Specifically, we show $\xi \leq 0$ when $A \geq 0.7$. Note that $\frac{\partial \xi}{\partial A} = -2\Delta p$, which is non-positive when $p_1\geq p_2$. It suffices to show that $\xi \leq 0$ for any $(p_1,p_2) \in G$ when $A = 0.7$. According to Proposition \ref{prop: P's property}, the set $P$ shrinks as $A$ increases, we have $G \subset Q:=\{(p_1,p_2) \in \underline{P}: 0.2 \leq p_1 \leq 0.5, 0.2 \leq p_2 \leq 0.5\}$.


From the above, it suffices to prove that $\xi \leq 0$ for any $(p_1,p_2)\in Q$ when $A = 0.7$. For $\xi$ we have the following first-order conditions:
\[
\begin{aligned}
    \frac{\partial \xi}{\partial p_1} &= 9p_2 - 14p_1 - 2A - 16p_1p_2 + 15p_1^2 + 3 = 0\\
    \frac{\partial \xi}{\partial p_2} &= - 8p_1^2 + 9p_1 + 2A - 4p_2 - 3 = 0
\end{aligned}
\]

When $A = 0.7$, the only real solution is $(p_1,p_2) \approx (0.8349,0.8436) \not \in Q$, which suggests the $\xi$ attains maximum on the boundary of $Q$. When $p_1= p_2$, on the edge between $(0.2,0.2)$ and $(0.5,0.5)$, $\xi = -3p_1^3 < 0$. 

When $p_2 = 0.2$, on the edge between $(0.2,0.2)$ and $(0.5,0.2)$, $\xi' = 
5p_1^3 - \frac{43}{5}p_1^2 + \frac{17}{5}p_1 - \frac{2}{5}$ and $\frac{d \xi'}{d p_1} = 15p_1^2 - \frac{83}{5}p_1 + \frac{17}{5}$. There are two roots for $\frac{d \xi'}{d p_1} = 0$: $p_1' = \frac{83}{150} \pm \frac{\sqrt{1789}}{150}$. Only the smaller root is on the edge, at which $\xi' \approx -0.107$. It implies $\xi$ is negative on the edge between $(0.2,0.2)$ and $(0.5,0.2)$.

When $p_1 = 0.5$, on the edge between $(0.5,0.2)$ and $(0.5,0.5)$, $\xi'' = 
-2p_2^2 + \frac{9}{10}p_2 - \frac{13}{40}
$. $\xi''$ is maximized at $p_2 =\frac{9}{40}$, at which $\xi'' = -\frac{179}{800} < 0$. 

The above analysis demonstrates that $\xi < 0$ in $Q$ and thus in $G$ when $A \geq 0.7$, which finishes our proof.
\end{proof}




\begin{claim}
    For any point $(p_1, p_2)$ in the $Arc(\underline{M}, M^{\text{low}})$, $\bar{\mu}^{SW}(\underline{p}, \underline{p})  \geq \bar{\mu}^{SW}(p_1, p_2)$.
\end{claim}

\begin{proof}
In $Arc(\underline{M}, M^{\text{low}})$, $-1 \leq \frac{\mathrm{d}p_2}{\mathrm{d}p_1} \leq 0$.

\[
\mathrm{d} \bar{\mu}^{SW}(p_1, p_2) = \frac{\partial \bar{\mu}^{SW}}{\partial p_1}\mathrm{d}p_1 + \frac{\partial \bar{\mu}^{SW}}{\partial p_2} \mathrm{d}p_2
\]
where the following cases apply:
\[
\frac{\mathrm{d}\bar{\mu}^{SW}(p_1, p_2)}{\mathrm{d}p_1}
\begin{cases}
  \leq \frac{\partial \bar{\mu}^{SW}}{\partial p_1}, \quad \text{if } \frac{\partial \bar{\mu}^{SW}}{\partial p_2} \geq 0, \\
  <  \frac{\partial \bar{\mu}^{SW}}{\partial p_1} - \frac{\partial \bar{\mu}^{SW}}{\partial p_2}, \quad \text{if } \frac{\partial \bar{\mu}^{SW}}{\partial p_2} < 0.
\end{cases}
\]

As a direct consequence of Claim \ref{negative partial derivative}, we have \( \frac{\mathrm{d}\bar{\mu}^{SW}(p_1, p_2)}{\mathrm{d}p_1} \leq 0 \) when \( \frac{\partial J^{SW}}{\partial p_2} \geq 0 \).

When \( \frac{\partial \bar{\mu}^{SW}}{\partial p_2} < 0 \), it suffices to show that for any point \( (p_1, p_2) \in Arc(\underline{M}, M^{\text{low}}) \backslash (\underline{p}, \underline{p}) \),
\[
\frac{\partial \bar{\mu}^{SW}}{\partial p_1} - \frac{\partial \bar{\mu}^{SW}}{\partial p_2} < 0.
\]

Note that
\[
\begin{aligned}
    \frac{\partial \bar{\mu}^{SW}}{\partial p_1} - \frac{\partial \bar{\mu}^{SW}}{\partial p_2} &= \left(\frac{\partial \underline{\alpha}}{\partial p_1} - \frac{\partial \underline{\alpha}}{\partial p_2}\right)(SW_1 - SW_2) + 2(p_1 - p_2)(p_1 - \underline{\alpha} p_1 + \underline{\alpha} p_2 - 1)\\
    &\leq \left(\frac{\partial \underline{\alpha}}{\partial p_1} - \frac{\partial \underline{\alpha}}{\partial p_2}\right)(SW_1 - SW_2) + 2(p_1 - p_2)(p_1 - 1).
\end{aligned}
\]

Taking the partial derivative of IC1 with respect to \(p_2\):
\[
\begin{aligned}
    \frac{\partial \max_{p_1'} p_1' D_1^2(p_1', p_2) }{\partial p_2} &= p_1\frac{\partial D_1^2}{\partial p_2} + \frac{\partial \underline{\alpha}}{\partial p_2} p_1 (D_1^1 - D_1^2) + \underline{\alpha} p_1 \left(\frac{\partial D_1^1}{\partial p_2} - \frac{\partial D_1^2}{\partial p_2}\right).
\end{aligned}
\]

Let \(p_1^*\) be the maximizer of \(p_1^* D_1^2(p_1^*, p_2)\). By the envelope theorem,
\[
\frac{\partial \max_{p_1'} p_1' D_1^2(p_1', p_2)}{\partial p_2} = \frac{\partial p_1 D_1^2(p_1, p_2)}{\partial p_2} \Big|_{p_1 = p_1^*} = p_1^*(1 - p_1^*).
\]

    After some algebra, we obtain:
\[
\begin{aligned}
    \frac{\partial \underline{\alpha}}{\partial p_2} &= \frac{p_1(1 - p_1) + \underline{\alpha} p_1 (p_1 - p_2) - p_1^*(1 - p_1^*)}{p_1\left(\frac{1}{2}(p_1 - p_2)^2 - (1 - A)^2\right)}\\
    &\leq \frac{p_1(1 - p_1) - \frac{1}{4}}{p_1\left(\frac{1}{2}(p_1 - p_2)^2 - (1 - A)^2\right)} \quad (\underline{\alpha} \geq 0, 0 \leq p_1^* \leq \frac{1}{2})\\
\end{aligned}
\]

According to the previous analysis we already have 

\[\frac{\partial \underline{\alpha}}{\partial p_1} \geq \frac{\frac{A^2}{2} - A + p_2 + p_1^2-2p_1 + 1}{p_1[\frac{1}{2}(p_1 - p_2)^2 - (1 - A)^2]}\]

Thus,
\[
\begin{aligned}
    \frac{\partial \underline{\alpha}}{\partial p_1} - \frac{\partial \underline{\alpha}}{\partial p_2} &\geq \frac{\frac{A^2}{2} - A + p_2 + p_1^2 + 1 - 2p_1 - p_1(1 - p_1) + \frac{1}{4}}{p_1(\frac{1}{2}(p_1 - p_2)^2 - (1 - A)^2)}\\
    &\geq \frac{\frac{A^2}{2} - A + 2p_1^2 - 3p_1 + p_2 + \frac{5}{4}}{- \frac{1}{2} p_1 (1 - A)^2}. \quad (|p_1 - p_2| \leq 1-A)\\
    & = \frac{A^2 - 2A + 4p_1^2 - 6p_1 + 2p_2 + \frac{5}{2}}{-p_1(1-A)^2}
\end{aligned}
\]

The second inequality comes from the fact that $A^2 - 2A + 4p_1^2 - 6p_1 + 2p_2 + \frac{5}{2}:=\tau > 0$ in $Arc(\underline{M}, M^{\text{low}})$ when $A \geq 0.7$. To see why, notice that the $\frac{\partial \tau}{\partial A} = 2A - 2\leq 0$. It suffices to show $\tau > 0$ when $A = 0.7$. Since in $Arc(\underline{M}, M^{\text{low}})$, $0.2 < p_1,p_2 < 0.45$, we have $\tau \geq 0.05 > 0$. 

Therefore we have
\[
\begin{aligned}
    \frac{\partial \bar{\mu}^{SW}}{\partial p_1} - \frac{\partial \bar{\mu}^{SW}}{\partial p_2} &\leq   \frac{A^2 - 2A + 4p_1^2 - 6p_1 + 2p_2 + \frac{5}{2}}{-p_1(1-A)^2} \left(-\frac{(p_1 - p_2)^3}{3}\right) + 2(p_1 - p_2)(p_1 - 1) \\
    &\leq \frac{(A^2 - 2A + 4p_1^2 - 6p_1 + 2p_2 + \frac{5}{2})(p_1 - p_2)}{3p_1} + 2(p_1 - p_2)(p_1 - 1) \quad (|p_1 - p_2| \leq 1 - A)\\
    &= \frac{\Delta p(A^2 -2A+10p_1^2-12p_1 + 2p_2 + \frac{5}{2})}{3p_1}\\
\end{aligned}
\]

Define $\omega: =A^2 -2A+10p_1^2-12p_1 + 2p_2 + \frac{5}{2}$. Since $\frac{\partial \omega}{\partial A} = (2A-2)\Delta p \leq 0$ when $p_1\geq p_2$, it suffices to show that $\omega \leq 0$ when $A = 0.7$. When $0.2 < p_2 \leq p_1 < 0.45$,
\[\left(\frac{\partial \omega}{\partial p_1}, \frac{\partial \omega}{\partial p_2}\right) = (20p_1-12, 2) \neq 0.\]

Thus, $\omega$ attains its maximum on the boundary of $Q' = \{(p_1,p_2): 0.2 \leq p_2\leq p_1\leq 0.45\}$. When $p_1 = p_2$, on the edge between $(0.2,0.2)$ and $(0.45,0.45)$, 
$$
\omega' =A^2 -2A+10p_1^2-10p_1 + \frac{5}{2} = 10p_1^2 - 10p_1 + 1.59 < 0.
$$

When $p_2 = 0.2$, on the edge between $(0.2,0.2)$ and $(0.45,0.2)$, 
$$
\omega'' = A^2 -2A+10p_1^2-12p_1 + 2.9
$$

Since $\frac{\partial \omega}{\partial p_1} = 20p_1-12 < 0$, $\omega''$ is maximized at $p_1 = 0.2$, where $\omega'' = 0.01 < 0$. 

When $p_1 = 0.45$, on the edge between $(0.45,0.2)$ and $(0.45,0.45)$, we have

\[\omega''' = 2p_2 - \frac{357}{200} < 0
\]

We thus conclude $\omega < 0$, which implies $\frac{\partial \bar{\mu}^{SW}}{\partial p_1} - \frac{\partial \bar{\mu}^{SW}}{\partial p_2} < 0$. Hence the claim is proved. 

The above analysis suggests that in $\underline{P}$, the social welfare is maximized at point $(\underline{p},\underline{p})$. Due to the symmetry of our model we can obtain the same result for $\bar{P}$ by changing the role of $p_1$ and $p_2$. Therefore, we conclude that the social welfare is maximized at point $(\underline{p},\underline{p})$. 

Because $J^{CS} = J^{SW} - J^{\Pi}$, we immediately know that the consumer surplus is maximized at point $(\underline{p}, \underline{p})$.
\end{proof}

\subsection*{Proof of Observation \ref{observation:corner}}

The first item is immediate by plugging $\alpha = 1$ into $\widehat{\text{IC}1}$ and $\widehat{\text{IC}2}$ and plugging $\alpha = 0$ into $\widehat{\text{IC}1'}$ and $\widehat{\text{IC}2'}$. 

For the second item, we define a new correspondence $\hat{\varphi}(p_1,p_2)$: 

\[
\hat{\varphi}(p_1,p_2) = \left\{\alpha \in [0,1]: 
\begin{cases}
    \text{IC1 and IC2}, &p_1 < A, p_2<A\\ 
    \widehat{\text{IC}1} \text{ and } \widehat{\text{IC}2}, &p_1 \geq A, p_2 < A\\
    \widehat{\text{IC}1'} \text{ and } \widehat{\text{IC}2'}, &p_1 < A, p_2 \geq A\\
    \widehat{\text{IC}1''} \text{ and } \widehat{\text{IC}2''}, &p_1 \geq A, p_2 \geq A\\
\end{cases}\right\}.
\]

Now we show that for any $(p_1,p_2)\in P$, $\hat{\varphi}(p_1,p_2) \subseteq {\varphi}(p_1,p_2)$. It suffices to consider the nondegenerate cases, i.e., when at least one of $p_1$ and $p_2$ is no less than A. 

When $p_1 \geq A, p_2<A$, we have 

\[\frac{\max_{p_1'}\pi_1^2(p_1',p_2)}{\pi_1^1} \leq \hat{\alpha} \leq 1- \frac{\max_{p_2'} \pi_2^2(p_1,p_2') -\pi_2^2}{\hat{\pi}_2^1-\pi_2^2} \]

We first notice that 
$$
\frac{\max_{p_1'}\pi_1^2}{\pi_1^1} 
\geq \frac{\max_{p_1^\prime} \pi_{1}^2-\pi_{1}^2}{p_1 B} = \frac{\max_{p_1^\prime} \pi_{1}^2-\pi_{1}^2}{\pi_1^1 - \pi_1^2}.
$$

Also, since 
$$
\begin{aligned}
    D_2^1 &=1-F(A-\Delta p)+ \int_{p_{1}}^{A} F(u)f(u-\Delta p)du \\
    &\geq 1 - F(A-p_1+p_2)\\
    & \geq 1 - F(p_2) = \hat{D}_2^1 \quad (p_1\geq A),
\end{aligned}
$$
we have $\hat \pi_2^1 \leq \pi_2^1$. Therefore, 

\[1- \frac{\max_{p_2'} \pi_2^2(p_1,p_2') -\pi_2^2}{\hat{\pi}_2^1-\pi_2^2} 
\leq 1- \frac{\max_{p_2'} \pi_2^2(p_1,p_2') -\pi_2^2}{{\pi}_2^1-\pi_2^2} = \bar{\alpha}. \]

Thus, when $p_1 \geq A$, $p_2 < A$, $\hat{\varphi}(p_1,p_2) \subseteq \varphi(p_1,p_2)$. We can obtain two other inequalities by changing the role of $p_1$ and $p_2$, which prove the case when $p_1< A$ and $p_2\geq A$.

When $p_1,p_2\geq A$, we have

\[\frac{\max_{p_1'}\pi_1^2}{\hat{\pi}_1^1}\leq \hat{\alpha} \leq 1 - \frac{\max_{p_2'}\pi_2^2}{\hat{\pi}_2^1}\]

We know that  
\[\frac{\max_{p_1'}\pi_1^2}{\hat{\pi}_1^1} \geq \frac{\max_{p_1'}\pi_1^2}{{\pi}_1^1} \geq \frac{\max_{p_1^\prime} \pi_{1}^2-\pi_{1}^2}{p_1 B} = \frac{\max_{p_1^\prime} \pi_{1}^2-\pi_{1}^2}{\pi_1^1 - \pi_1^2},\]
\[\frac{\max_{p_2'}\pi_2^2}{\hat{\pi}_2^1} \geq \frac{\max_{p_2'}\pi_2^2}{{\pi}_2^1} \geq \frac{\max_{p_2^\prime} \pi_{2}^2-\pi_{2}^2}{p_2 B} = \frac{\max_{p_2^\prime} \pi_{2}^2-\pi_{2}^2}{\pi_2^1 - \pi_2^2}.\]

We thus proved that for any $(p_1,p_2) \in P$, $\hat{\varphi}(p_1,p_2) \subseteq \varphi(p_1,p_2)$, which implies $\hat{P}\subseteq P$.

\bibliographystyle{aea}
\bibliography{Search}

\end{document}